\documentclass[12pt]{amsart}

\usepackage{amsmath}
\usepackage{amsfonts}
\usepackage{amssymb}
\usepackage[all]{xy}           

\usepackage{bbding}
\usepackage{txfonts}
\usepackage{amscd}
\usepackage{dsfont}
\usepackage{mathrsfs}

\usepackage{yfonts}

\usepackage[shortlabels]{enumitem}
\usepackage{ifpdf}
\ifpdf
  \usepackage[colorlinks,final,backref=page,hyperindex]{hyperref}
\else
  \usepackage[colorlinks,final,backref=page,hyperindex,hypertex]{hyperref}
\fi
\usepackage{tikz}
\usepackage[active]{srcltx}

\usepackage{graphicx}
\makeatletter

\newtheorem{thm}{Theorem}[section]
\newtheorem{lem}[thm]{Lemma}
\newtheorem{cor}[thm]{Corollary}
\newtheorem{pro}[thm]{Proposition}
\theoremstyle{definition}   
\newtheorem{ex}[thm]{Example}
\newtheorem{rmk}[thm]{Remark}
\newtheorem{defi}[thm]{Definition}

\newcommand {\emptycomment}[1]{}

\newcommand{\be }{\begin{equation}}
\newcommand{\ee }{\end{equation}}

\newcommand{\huaL}{\mathcal{L}}

\newcommand{\huaE}{\mathcal{E}}

\newcommand{\huaV}{\mathcal{V}}

\newcommand{\huaC}{{\mathcal{C}}}
\newcommand{\huaD}{\mathcal{D}}

\newcommand{\huaH}{\mathcal{H}}

\newcommand{\frke}{\mathfrak e}

\newcommand{\frks}{\mathfrak s}

\newcommand{\Id}{\rm{Id}}

\newcommand{\br}[1]{   [ \cdot,    \cdot  ]   }

\newcommand{\Hom}{\mathrm{Hom}}

\newcommand{\gl}{\mathfrak {gl}}

\newcommand{\ad}{\mathrm{ad}}

\newcommand{\K}{\mathbb{K}}

\newcommand{\R}{\mathbb{R}}
\newcommand{\C}{\mathbb{C}}
\newcommand{\BH}{\mathbb{H}}
\newcommand{\BO}{\mathbb{O}}

\begin{document}

\title[The Yang-Baxter equation, Leibniz algebras,  racks]{The Yang-Baxter equation, Leibniz algebras,  racks and related algebraic structures}


\author{Nanyan Xu}
\address{Department of Mathematics, Jilin University, Changchun 130012, Jilin, China}
\email{xuny20@mails.jlu.edu.cn}

\author{Yunhe Sheng}
\address{Department of Mathematics, Jilin University, Changchun 130012, Jilin, China}
\email{shengyh@jlu.edu.cn}


\begin{abstract}
The purpose of this paper is to clarify the relations between various constructions of solutions of the Yang-Baxter equation from Leibniz algebras, racks, 3-Leibniz algebras, 3-racks, linear racks, trilinear racks, and give new constructions of solutions of the Yang-Baxter equation. First we show that a 3-Leibniz algebra naturally gives rise to a 3-rack on the underlying vector space, which generalizes Kinyon's construction of racks from Leibniz algebras. Then we show that a trilinear rack naturally gives rise to a linear rack. Combined with Lebed's construction of solutions of the Yang-Baxter equation from linear racks, our results give an intrinsic explanation of Abramov and Zappala's construction of solutions of the Yang-Baxter equation from trilinear racks. Next we show that a 3-Leibniz algebra gives rise to a trilinear rack, which generalizes Abramov and Zappala's construction from 3-Lie algebras. Finally, we construct solutions of the Yang-Baxter equation using central extensions of 3-Leibniz algebras and Leibniz algebras. In particular, given a 3-Leibniz algebra, there are two different approaches to construct solutions of the Yang-Baxter equation, namely either consider the central extension of the Leibniz algebra on the fundamental objects, or consider the Leibniz algebra on the fundamental objects of the central extension of the 3-Leibniz algebra. We also show that there is a homomorphism between the corresponding solutions.
\end{abstract}

\renewcommand{\thefootnote}{}

\maketitle

\tableofcontents

\allowdisplaybreaks


\section{Introduction}
\label{sec:intr}

\subsection{The Yang-Baxter equation}
In 1967, Yang discovered the Yang-Baxter equation as the consistency condition for the factorization in the quantum mechanical many-body problem \cite{Yang}. At the same time, the Yang-Baxter equation also appeared in Baxter's work of  eight-vertex model \cite{Baxter}. A solution of the Yang-Baxter equation on a vector space $V$ is an  invertible  linear map $R:V\otimes V\rightarrow V\otimes V$ satisfying
\begin{eqnarray}\label{YBE}
(R\otimes{\Id}_V)({\Id}_V\otimes R)(R\otimes{\Id}_V)=({\Id}_V\otimes R)(R\otimes{\Id}_V)({\Id}_V\otimes R).
\end{eqnarray}
 The Yang-Baxter equation is widely used in various areas of mathematics and physics, such as integrable systems, knot theory, statistical mechanics, quantum field theory, quantum computing, braided categories, etc.
FRT-construction, introduced by Faddeev, Reshetikhin and Takhtajan, provided a way to obtain some bialgebras by means of $R$-matrices, where $R$-matrices are viewed as matrix solutions of the Yang-Baxter equation \cite{FRT}.
Jones proposed ``baxterization'' which shows the connection between representations of the braid group and solutions of the Yang-Baxter equation \cite{Jones}.

In order to simplify the study of the Yang-Baxter equation, Drinfel'd proposed to consider the Yang-Baxter equation at the level of sets \cite{Drinfeld}. 
There are a lot of research results on the set-theoretical solutions of the Yang-Baxter equation. Etingof-Schedler-Soloviev \cite{ESS} and Lu-Yan-Zhu \cite{LYZ} paid attention to the non-degenerate involutive set-theoretical solutions and studied the structure group,   which is used to classify set-theoretical solutions. 
 In 2006, Rump introduced the notion of braces, which are generalization of radical rings and give  rise to non-degenerate involutive set-theoretical solutions of the Yang-Baxter equation \cite{Rump}. Later, Guarnieri and Vendramin extended braces to the noncommutative setting by defining skew braces in order to obtain non-degenerate non-involutive set-theoretical solutions \cite{GV2}. Bardakov and Gubarev obtained non-degenerate set-theoretical solutions of the Yang-Baxter equation using the connection between Rota-Baxter groups \cite{GLS} and skew-left braces \cite{BG}. In \cite{BGST}, Bai, Guo, Sheng and Tang introduced the notion of post-groups, which are equivalent to skew left braces, and give rise to   set-theoretical solutions of the Yang-Baxter equation.

\subsection{Racks, Leibniz algebras, 3-racks and 3-Leibniz algebras}
 Set-theoretical solutions of the Yang-Baxter equation are closely related to many algebraic structures as mentioned above. It is worth noting that there is another algebraic structure, named rack, which also plays an important role in the study of the Yang-Baxter equation. As an algebraic interpretation of Reidemeister moves, a {\bf rack} is a set $X$ with a binary operator $\lhd:X\times X\to X$ satisfying
\begin{itemize}
  \item {\rm self-distributive law:} $(x\lhd y)\lhd z=(x\lhd z)\lhd(y\lhd z),\,\,\,\forall x,y,z\in X$,
  \item for all $x\in X$, the map $\cdot\lhd x:X\rightarrow X$ sending $y\in X$ to $y\lhd x$ is a bijection.
\end{itemize}\noindent
If the binary operation $\lhd$ only satisfies the self-distributive law, then $(X,\lhd)$ is called a {\bf shelf}.
Racks are connected with the theories of groups, group presentations and crossed modules, and open the practical possibility of finding a complete sequence of computable invariants for framed links and for 3-manifolds. Many examples of racks are listed in \cite{FR}. In \cite{AG}, Andruskiewitsch and Gra${\rm\tilde{n}}$a proved that there is a one-to-one corresponding between racks and certain set-theoretical solutions of the Yang-Baxter equation (Theorem \ref{rack-set-YBO}).  Biyogmam generalized racks to $n$-racks and gave some examples in \cite{Biyogmam}. In this paper, we consider the case $n=3$, i.e. $3$-racks.  One can construct a $3$-rack $(X,T^\lhd)$ by a rack $(X,\lhd)$ \cite{EGM}
and  one can also obtain a rack $(X\times X,\lhd^T)$ from a $3$-rack $(X,T)$ \cite{ESZ}.

Racks are closely related to Leibniz algebras, which were introduced by Loday in his study of algebraic K-theory \cite{Loday1}. Specially, a {\bf Leibniz algebra} is a vector space $\huaE$ together with a bilinear operation $[\cdot,\cdot]_\huaE$ such that
\begin{eqnarray*}
\label{Leibniz}[[x,y]_\huaE,z]_\huaE=[[x,z]_\huaE,y]_\huaE+[x,[y,z]_\huaE]_\huaE,
\,\,\,\forall x,y,z\in\huaE.
\end{eqnarray*}
In \cite{Kinyon}, Kinyon proved that a Leibniz algebra $(\huaE,[\cdot,\cdot]_\huaE)$ gives rise to a rack structure $\lhd$ on $\huaE$ via  $x\lhd y=\exp(\ad^R_y)(x)$ (Theorem \ref{lei to rack thm}).

\emptycomment{Noting that Leibniz algebras widely appear in many areas of mathematics, such as differential geometry, homological algebra, classical algebraic topology and non-commutative geometry.
}Furthermore, Casas, Loday and Pirashvili defined $n$-Leibniz algebras and proved that an $n$-Leibniz algebra gives rise to a Leibniz algebra structure on the set of fundamental objects \cite{Casas}. Biyogmam  proved that the tangent space of a Lie $n$-rack at the neutral element has an $n$-Leibniz algebra structure \cite{Biyogmam}. Since $3$-Leibniz algebras and $3$-racks are generalizations of Leibniz algebras and racks respectively, the following question naturally arises:
\begin{itemize}
\item[{\bf Q:}]Is there a $3$-rack structure on a $3$-Leibniz algebra as a generalization of Kinyon's result?
\end{itemize}

On the other hand,   Lebed showed that a central Leibniz algebra $(\huaE,[\cdot,\cdot]_\huaE,{\bf 1})$ gives rise to a solution of the Yang-Baxter equation via the following formula \cite{Lebed1}:
 \begin{eqnarray*}
R^{Lei}(x\otimes y)=y\otimes x+{\bf 1}\otimes[x,y]_\huaE,\,\,\,\forall x,y\in\huaE,
\end{eqnarray*}
Due to the close relationship between Leibniz algebras and $3$-Leibniz algebras, it is natural to consider the following question:
\begin{itemize}
  \item[{\bf Q:}]Can $3$-Leibniz algebras give rise to solutions of the Yang-Baxter equation?
\end{itemize}

\subsection{Linear racks and trilinear racks}

Kr${\rm \ddot{a}}$hmer and Wagemann defined linear shelfs as a coalgebraic version of shelfs \cite{Krahmer}. Similar definition has been proposed in \cite{Carter} and \cite{Lebed3}. In this paper, we use linear racks to express the coalgebraic version of racks so that the relationship between linear racks and linear shelfs is similar to the relationship between racks and shelfs.
Abramov and Zappala generalized linear racks to reversible TSD, which we call trilinear racks in this article. They showed that trilinear racks can be used to construct solutions of the Yang-Baxter equation \cite{Abramov}. On the other hand,   Lebed showed that linear racks give rise to solutions of the Yang-Baxter equation in \cite{Lebed3}. However, the relation between the above two approaches of constructing solutions of the Yang-Baxter equation is not clear. Since at the level of sets, there is a passage from $3$-racks to racks, so it is natural to consider the following question:
\begin{itemize}
\item[{\bf Q:}]Is there a passage from trilinear racks to linear racks as a coalgebraic version of the passage from $3$-racks to racks? If so, is the construction of solutions of the Yang-Baxter equation from trilinear racks given in \cite{Abramov} is consistent with the one given in  \cite{Lebed3}?
\end{itemize}

In \cite{Abramov}, the authors also construct  trilinear racks from   $3$-Lie algebras. While Lebed show that one can construct a linear rack from a Leibniz algebra in \cite{Lebed3}. Since $3$-Leibniz algebras are natural generalizations of Leibniz algebras, so we  propose the following question:
\begin{itemize}
\item[{\bf Q:}]Whether trilinear rack structures can be constructed using $3$-Leibniz algebras?
\end{itemize}

\subsection{Main results and outline of the paper}

The purpose of this paper is to solve the questions proposed above, and clarify the relations between various constructions of solutions of the Yang-Baxter equation. First we show that a 3-Leibniz algebra naturally gives rise to a 3-rack structure on the underlying vector space, which generalizes Kinyon's construction of racks from Leibniz algebras given in \cite{Kinyon}. Therefore, there are two approaches to obtain a rack from a 3-Leibniz algebra, and we show that there is a rack morphism between those two racks, which is a map preserves rack structure. If we further consider the corresponding solutions of the Yang-Baxter equation, there is also a homomorphism between solutions. We summarize these results in the following commutative diagram, and the colorful arrow are what we obtain in this paper. This is the content in Section \ref{sec-3-Lei-and-3-rack}.\vspace{-5mm}
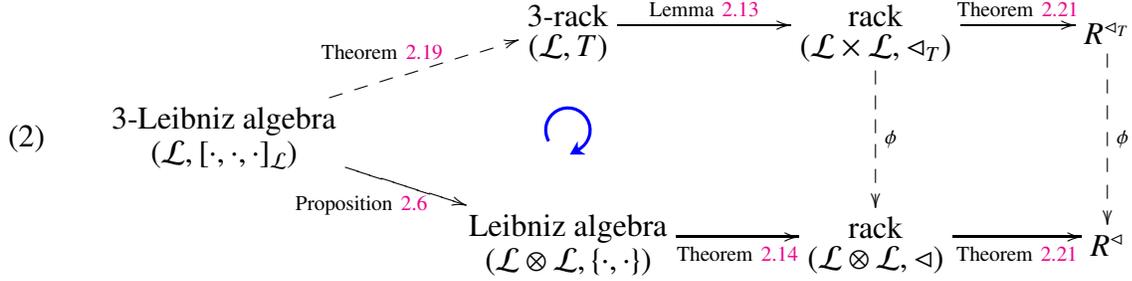
\begin{center}
\begin{equation}\label{diagram:sec2}
\begin{array}{l}
\xymatrix@C=8ex@R=1.6ex{
  &\txt{\rm $3$-rack\\$(\huaL,T)$}
    \ar[r]^-{{\rm Lemma}~\ref{3-rack-to-rack}}
  &\txt{\rm rack\\$(\huaL\times\huaL,\lhd_{T})$}
    {\color{blue} \ar@[blue]@{-->}[dd]^-{\phi}}\ar[r]^-{{\rm Theorem}~\ref{rack-set-YBO}}
  &\txt{$R^{\lhd_{T}}$}
  \ar@[blue]@{-->}[dd]^-{\phi}  \\
  \txt{\rm $3$-Leibniz algebra\\$(\huaL,[\cdot,\cdot,\cdot]_\huaL)$}
    \ar@[blue]@{-->}[ru]^-{{\rm Theorem}~\ref{from-3-Lei-to-3-rack-thm}\quad}
    \ar[dr]_-{\rm Proposition~\ref{3-lei-to-lei}\quad}
  &\rotatebox{165}{\color{blue}{\txt{\Huge $\circlearrowright$}}}
  &&\\
  &\txt{\rm Leibniz algebra\\$(\huaL\otimes\huaL,\{\cdot,\cdot\})$}
    \ar[r]_-{{\rm Theorem}~\ref{lei to rack thm}}
  &\txt{\rm rack\\$(\huaL\otimes\huaL,\lhd)$}
 \ar[r]_-{{\rm Theorem}~\ref{rack-set-YBO}}
  &\txt{$R^{\lhd}$}\\
}
\end{array}
\end{equation}
\end{center}

Then we show that a trilinear rack naturally gives rise to a linear rack, which generalizes the constructions from a 3-Leibniz algebra to a Leibniz algebra \cite{Casas} and from a 3-rack to a rack \cite{ESZ}. Then combine with  Lebed's construction of solutions of the Yang-Baxter equation from linear racks give   in \cite{Lebed3}, a trilinear rack naturally gives rise to a solution  of the Yang-Baxter equation. This alternatively gives an intrinsic explanation of the constructions given by Abramov and Zappala in \cite{Abramov}, as explained by the following diagram. This is the content of Section \ref{sec-lin-rack-tri-rack}.\vspace{-5mm}
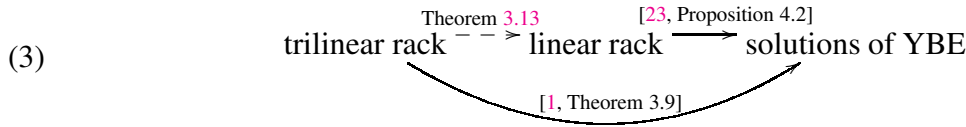
\begin{center}
\begin{equation}\label{diagram:sec3}
\begin{array}{l}
\hspace{1mm}\xymatrix@R=0.1pc{
 \txt{trilinear rack}\ar@[blue]@{-->}[r]^{\text{Theorem \ref{TSD-to-quan-rack}}} \ar@/_{3pc}/[rr]!U^(.4){\quad\quad\text{\cite[Theorem~3.9]{Abramov}}}& \txt{linear rack} \ar@{>}[r]^{\text{\cite[Proposition~4.2]{Lebed3}}}& \txt{solutions of YBE} \\
 &  &
 }
 \end{array}
 \end{equation}
 \end{center}

 Motivated from Lebed's construction of linear racks from  Leibniz algebras given in \cite{Lebed3}, we show that a 3-Leibniz algebra $\huaL$
naturally gives rise to a trilinear rack structure on $\K\oplus\huaL$. On the one hand, this generalizes the construction of trilinear racks from 3-Lie algebras given in  \cite{Abramov}. On the other hand, this makes it possible to construct solutions of the Yang-Baxter equation from a 3-Leibniz algebra directly. This is the content of Section \ref{sec-3-Lei-tri}.

Finally in Section \ref{sec-lei-3-lei},
we show that central 3-Leibniz algebras give rise to central Leibniz algebras, which lead to solutions of the Yang-Baxter equation. 
Given a 3-Leibniz algebra $\huaL$, there are two approaches to obtain a central Leibniz algebra, which lead to two different solutions of the Yang-Baxter equation. One approach is to consider the trivial central extension of the Leibniz algebra on the fundamental objects, while the other approach is to consider the Leibniz algebra on the fundamental objects of the trivial central extension of the 3-Leibniz algebra. We show that there is a homomorphism between the corresponding solutions of the Yang-Baxter equation, as illustrated by the following commutative diagram where $\overline{\huaL}=\K\oplus\huaL$:\vspace{-5mm}
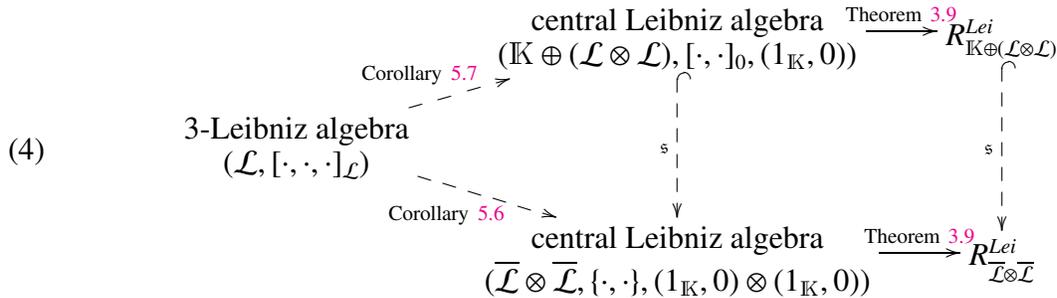
\begin{center}
\begin{equation}\label{diagram:sec5}
\begin{array}{l}
\xymatrix@C=4ex@R=2ex{
  &\txt{central Leibniz algebra\\ $(\mathds{K}\oplus(\huaL\otimes\huaL),[\cdot,\cdot]_0,(1_\mathds{K},0))$}
  \ar[r]^-{{\rm Theorem}~\ref{lei to solution}} \ar@[blue]@{^{(}-->}[dd]_-{\frks}
  &\txt{$R^{Lei}_{\mathds{K}\oplus(\huaL\otimes\huaL)}$}
  \ar@[blue]@{^{(}-->}[dd]_-{\frks}\\
  \txt{$3$-Leibniz algebra\\$(\huaL,[\cdot,\cdot,\cdot]_\huaL)$}
  \ar@[blue]@{-->}[dr]_-{{\rm Corollary}~\ref{3-lei-cen-ext-to-sol}\quad}
  \ar@[blue]@{-->}[ur]^-{{\rm Corollary}~\ref{3-lei-to-lei-and-sol-1}\quad}
  &&\\
  &\txt{central Leibniz algebra~\\
  $(\overline{\huaL}\otimes\overline{\huaL},\{\cdot,\cdot\},(1_\mathds{K},0)\otimes(1_\mathds{K},0))$}
  \ar[r]^-{{\rm Theorem}~\ref{lei to solution}}
  &\txt{$R^{Lei}_{\overline{\huaL}\otimes\overline{\huaL}}$}
}
\end{array}
\end{equation}
\end{center}

In this paper, we work over an algebraically closed field $\mathds{K}$ of characteristic 0 and all the vector spaces are over $\K$ and finite-dimensional.

\section{A passage from 3-Leibniz algebras to 3-racks}\label{sec-3-Lei-and-3-rack}

In this section, we construct a $3$-rack structure on a $3$-Leibniz algebra, which generalizes Kinyon's result in \cite{Kinyon}. Consequently, a 3-Leibniz algebra $\huaL$ gives rise to two rack structures on $\huaL\times\huaL$ and $\huaL\otimes \huaL$ respectively. We also show that there is a homomorphism between the corresponding set-theoretical solutions of the Yang-Baxter equation on $\huaL\times\huaL$ and $\huaL\otimes \huaL$.

\begin{defi}\cite{Casas}
A right {\bf $3$-Leibniz algebra} is a vector space $\huaL$ with a linear map $[\cdot,\cdot,\cdot]_\huaL:\huaL\otimes\huaL\otimes\huaL\rightarrow\huaL$ such that for any $x_1,x_2,x_3,y_1,y_2\in\huaL$, the following identity holds:
\begin{eqnarray}\label{3-Lei-equation}
&&[[x_1,x_2,x_3]_\huaL,y_1,y_2]_\huaL=[[x_1,y_1,y_2]_\huaL,x_2,x_3]_\huaL
+[x_1,[x_2,y_1,y_2]_\huaL,x_3]_\huaL+[x_1,x_2,[x_3,y_1,y_2]_\huaL]_\huaL.
\end{eqnarray}
In addition, if $[\cdot,\cdot,\cdot]_\huaL$ is skew-symmetric, then  $(\huaL,[\cdot,\cdot,\cdot]_\huaL)$ is a $3$-Lie algebra.
\end{defi}

\begin{defi}
Let $(\huaL_1,[\cdot,\cdot,\cdot]_{\huaL_1})$ and $(\huaL_2,[\cdot,\cdot,\cdot]_{\huaL_2})$ be two $3$-Leibniz algebras. A linear map $f:\huaL_1\rightarrow\huaL_2$ is a {\bf $3$-Leibniz algebra homomorphism} if
\begin{eqnarray*}
f[x,y,z]_{\huaL_1}=[f(x),f(y),f(z)]_{\huaL_2},\,\,\,\forall x,y,z\in\huaL_1.
\end{eqnarray*}
Moreover, $f$ is a {\bf $3$-Leibniz algebra isomorphism} if $f$ is bijective.
\end{defi}

\begin{rmk}
It is obvious that, given a right $3$-Leibniz algebra $(\huaL,[\cdot,\cdot,\cdot]_\huaL)$, then $(\huaL,[\cdot,\cdot,\cdot]'_\huaL)$ is a left $3$-Leibniz algebra, where $[x_1,x_2,x_3]'_\huaL=[x_3,x_2,x_1]_\huaL$ for all $x_1,x_2,x_3\in\huaL$. In this paper, we only consider right $3$-Leibniz algebras.
\end{rmk}

\begin{ex}\cite{Nambu,Yam}\label{ex-octonion}
Let $\BO$ be the octonions, one of the four normed division algebras (the
other three are $\R$, $\C$ and $\BH$), which can be viewed as an $8$-dimensional vector space over $\R$ with a basis $\{e_0,e_1,e_2,e_3,e_4,e_5,e_6,e_7\}$, where $e_0$ is  identified with the real number ${\rm 1}$. The multiplication of $\BO$ is non-associative and non-commutative, which is defined as follows:
\begin{center}
\begin{tabular}{|c|c|c|c|c|c|c|c|c|}
  \hline
  $e_ie_j$& $e_0$ & $e_1$ & $e_2$ & $e_3$ & $e_4$ & $e_5$ & $e_6$ & $e_7$\\
  \hline
  $e_0$ & $e_0$ & $e_1$ & $e_2$ & $e_3$ & $e_4$ & $e_5$ & $e_6$ & $e_7$\\
  \hline
   $e_1$ & $e_1$ & $-e_0$ & $e_4$ & $e_7$ & $-e_2$ & $e_6$ & $-e_5$ & $-e_3$\\
   \hline
   $e_2$ & $e_2$ & $-e_4$ & $-e_0$ & $e_5$ & $e_1$ & $-e_3$ & $e_7$ & $-e_6$\\
   \hline
  $e_3$ & $e_3$ & $-e_7$ & $-e_5$ & $-e_0$ & $e_6$ & $e_2$ & $-e_4$ & $e_1$\\
  \hline
  $e_4$ & $e_4$ & $e_2$ & $-e_1$ & $-e_6$ & $-e_0$ & $e_7$ & $e_3$ & $-e_5$\\
  \hline
  $e_5$ & $e_5$ & $-e_6$ & $e_3$ & $-e_2$ & $-e_7$ & $-e_0$ & $e_1$ & $e_4$\\
  \hline
  $e_6$ & $e_6$ & $e_5$ & $-e_7$ & $e_4$ & $-e_3$ & $-e_1$ & $-e_0$ & $e_2$\\
  \hline
  $e_7$ & $e_7$ & $e_3$ & $e_6$ & $-e_1$ & $e_5$ & $-e_4$ & $-e_2$ & $-e_0$\\
  \hline
\end{tabular}
\end{center}
Define a linear map $[\cdot,\cdot,\cdot]:\BO\otimes\BO\otimes\BO\to\BO$ by $[x,y,z]=z(yx)-y(zx)+(xy)z-(xz)y+(yx)z-y(xz)$. Then $(\BO,[\cdot,\cdot,\cdot])$ is a $3$-Leibniz algebra, but not a $3$-Lie algebra.
\end{ex}

\begin{ex}\cite{BZ}\label{ex-nil-3-lei}
Let $\huaL$ be a $3$-dimensional vector space with a basis $\{e_1,e_2,e_3\}$.
Then $\huaL$ equipped with the following linear map $[\cdot,\cdot,\cdot]_\huaL:\huaL\otimes\huaL\otimes\huaL\to\huaL$ is a $3$-Leibniz algebra:
$$[e_2,e_3,e_3]_\huaL=e_1,\,\,\,\,\,[e_3,e_3,e_3]_\huaL=e_2.$$
\end{ex}

\begin{pro}\label{3-lei-to-lei}\cite{Casas}
Let $(\huaL,[\cdot,\cdot,\cdot]_\huaL)$ be a $3$-Leibniz algebra. Then $(\huaL\otimes\huaL,\{\cdot,\cdot\})$ is a Leibniz algebra, where the Leibniz bracket on the fundamental objects is defined by
\begin{eqnarray}\label{fundamental-element-formula}
&&\{x_1\otimes x_2,y_1\otimes y_2\}
=[x_1,y_1,y_2]_\huaL\otimes x_2+x_1\otimes[x_2,y_1,y_2]_\huaL,\,\,\,
\forall x_1\otimes x_2,y_1\otimes y_2\in\huaL\otimes\huaL.
\end{eqnarray}
\end{pro}

\begin{pro}\label{lei-to-3-lei}\cite{Casas}
Let $(\huaE,[\cdot,\cdot]_\huaE)$ be a Leibniz algebra. Then $(\huaE,[\cdot,\cdot,\cdot]_\huaE)$ is a $3$-Leibniz algebra, where $[\cdot,\cdot,\cdot]_\huaE$ is defined as follows:
$$[x,y,z]_\huaE=[x,[y,z]_\huaE]_\huaE,\quad\forall x,y,z\in\huaE.$$
\end{pro}

\begin{ex}\label{ex-omni-3-lie}
The omni-Lie algebra \cite{Alan,LSW} associated to a vector space $V$ is $\gl(V)\oplus V$ equipped with the following bracket operation $\{\cdot,\cdot\}$:
$$\{(A,u),(B,v)\}=(-[A,B],Bu),\quad\forall (A,u),(B,v)\in\gl(V)\oplus V,$$
which is a Leibniz algebra. Then by {\rm Proposition \ref{lei-to-3-lei}}, we obtain a $3$-Leibniz algebra $(\gl(V)\oplus  V,[\cdot,\cdot,\cdot])$, where
$$[(A,u ),(B,v ),(C,w )]=([A,[B,C]],-[B,C]u),$$
for any $(A,u),(B,v),(C,w)\in\gl(V)\oplus V$.
\end{ex}

\begin{defi}\cite{Biyogmam}
Let $X$ be a set and $T:X\times X\times X\rightarrow X$ a ternary map.
\begin{itemize}
  \item A {\bf $3$-shelf} is a pair $(X,T)$ satisfying the ternary self-distributive law:
  \begin{eqnarray*}\label{TSD-formula-set}
  &&T(T(x_{1},x_{2},x_{3}),y_{1},y_{2})
   =T(T(x_{1},y_{1},y_{2}),T(x_{2},y_{1},y_{2}),T(x_{3},y_{1},y_{2})),\,\forall x_{1},x_{2},x_{3},y_{1},y_{2}\in X.
  \end{eqnarray*}
  \item A {\bf $3$-rack} is a $3$-shelf $(X,T)$ such that for any $x,y\in X$, the map $T(\cdot,x,y):X\rightarrow X$ sending $z\in X$ to $T(z,x,y)$ is bijective.
\end{itemize}
\end{defi}

\begin{ex}\cite{Biyogmam}
  Let $(G,\cdot)$ be a group. Define a ternary map $T:G\times G\times G\rightarrow G$ by $T(g_1,g_2,g_3)={g_3}^{-1}{g_2}^{-1}{g_1}g_2g_3$ for any $g_1,g_2,g_3\in G$. Then $(G,T)$ is a $3$-rack.

  \end{ex}
  \begin{ex}\cite{Biyogmam}
  Let $M$ be a ${\rm Z_4}$-module. Define a ternary map $T:M\times M\times M\to M$ by $T(m_1,m_2,m_3)=m_1+2m_2+2m_3$ for any $m_1,m_2,m_3\in M$. Then $(M,T)$ is a $3$-rack.
\end{ex}

The following lemmas show the close relationship between racks and $3$-racks.
\begin{lem}\label{rack-to-3-rack}\cite{EGM}
Let $(X,\lhd)$ be a rack. Define a ternary map $T:X\times X\times X\rightarrow X$ by
\begin{eqnarray*}\label{rack-to-3-rack-formula}
T^\lhd(x,y,z)=x\lhd(y\lhd z),\,\,\,\forall x,y,z\in X.
\end{eqnarray*}
Then $(X,T^\lhd)$ is a $3$-rack.
\end{lem}

\begin{lem}\label{3-rack-to-rack}\cite{ESZ}
Let $(X,T)$ be a $3$-rack. Define a binary map $\lhd_{T}:(X\times X)\times(X\times X)\to X\times X$ by
\begin{eqnarray}\label{3-rack-to-rack-formula}
(x_1,x_2)\lhd_{T}(y_1,y_2)=
(T(x_1,y_1,y_2),T(x_2,y_1,y_2)),\,\,\,\forall (x_1,x_2),(y_1,y_2)\in X\times X.
\end{eqnarray}
Then $(X\times X,\lhd_{T})$ is a rack.
\end{lem}

There is naturally a rack structure on a Leibniz algebra.
\begin{thm}\cite{Kinyon}\label{lei to rack thm}
Let $(\huaE,[\cdot,\cdot]_\huaE)$ be a Leibniz algebra. For any $x\in\huaE$, define $\ad^R_x:\huaE\rightarrow\huaE$ by
\begin{eqnarray*}
\ad^R_xy=[y,x]_\huaE,\,\,\,\forall y\in\huaE.\label{lei-to-rack-formula}
\end{eqnarray*}
Then $(\huaE,\lhd)$ is a rack, where $x\lhd y=\exp(\ad^R_y)(x)$ and $\exp(\ad^R_y)=\sum\limits_{k=0}^{+\infty}\frac{(\ad^R_y)^k}{k!}$.
\end{thm}\noindent

The goal of this section is to construct a $3$-rack structure on a $3$-Leibniz algebra generalizing Theorem \ref{lei to rack thm}. We need some preparations.

\begin{defi}
Let $(\huaL,[\cdot,\cdot,\cdot]_\huaL)$ be a $3$-Leibniz algebra. A linear map $\huaD:\huaL\rightarrow\huaL$ is called a {\bf derivation} of $\huaL$ if it satisfies:
\begin{eqnarray*}
\huaD[x,y,z]_\huaL=[\huaD x,y,z]_\huaL+[x,\huaD y,z]_\huaL+[x,y,\huaD z]_\huaL,\,\,\,\forall x,y,z\in\huaL.
\end{eqnarray*}
\end{defi}

\begin{ex}
Let $(\huaL,[\cdot,\cdot,\cdot]_\huaL)$ be a $3$-Leibniz algebra. For any $x_1,x_2\in\huaL$, define a linear map $\ad^{R}_{x_1,x_2}:\huaL\rightarrow\huaL$ by
\begin{eqnarray*}
\ad^{R}_{x_1,x_2}(y)=[y,x_1,x_2]_\huaL,\,\,\,\forall y\in\huaL.
\end{eqnarray*}
It is obvious that $\ad^{R}_{x_1,x_2}$ is a derivation of $\huaL$ by \eqref{3-Lei-equation}.
\end{ex}

\begin{lem}\label{exp(D)}
Let $(\huaL,[\cdot,\cdot,\cdot]_\huaL)$ be a $3$-Leibniz algebra and $\huaD:\huaL\rightarrow\huaL$ a derivation. Then $\exp(\huaD)=\sum_{k=0}^{+\infty}\frac{1}{k!}\huaD^k$ is a $3$-Leibniz algebra isomorphism of $\huaL$.
\end{lem}

\begin{proof}
For any $x,y,z\in\huaL,$ the equality $\huaD[x,y,z]_\huaL=[\huaD x,y,z]_\huaL+[x,\huaD y,z]_\huaL+[x,y,\huaD z]_\huaL$ leads by induction to
\begin{eqnarray*}
\frac{1}{k!}\huaD^k[x,y,z]_\huaL=
\sum_{\substack{i_1+i_2+i_3=k\\ i_1,i_2,i_3\geqslant0}}
\Big[\frac{\huaD^{i_1}(x)}{i_{1}!},\frac{\huaD^{i_2}(y)}{i_{2}!},\frac{\huaD^{i_3}(z)}{i_{3}!}\Big]_\huaL.
\end{eqnarray*}
Then we obtain
\begin{eqnarray*}
[\exp(\huaD)(x),\exp(\huaD)(y),\exp(\huaD)(z)]_\huaL
&=&\Big[\sum_{i_1=0}^{+\infty}\frac{\huaD^{i_1}(x)}{i_1!},\sum_{i_2=0}^{+\infty}\frac{\huaD^{i_2}(y)}{i_2!},\sum_{i_3=0}^{+\infty}\frac{\huaD^{i_3}(z)}{i_3!}\Big]_\huaL\\
&=&\sum_{k=0}^{+\infty}\sum_{\substack{i_1+i_2+i_3=k\\ i_1,i_2,i_3\geqslant0}}
\Big[\frac{\huaD^{i_1}(x)}{i_1!},\frac{\huaD^{i_2}(y)}{i_2!},\frac{\huaD^{i_3}(z)}{i_3!}\Big]_\huaL\\
&=&\sum_{k=0}^{+\infty}\frac{1}{k!}\huaD^k[x,y,z]_\huaL\\
&=&\exp(\huaD)[x,y,z]_\huaL.
\end{eqnarray*}
Obviously, $\exp(-\huaD)$ is the inverse of $\exp(\huaD)$. Thus, $\exp(\huaD)$ is an isomorphism.
\end{proof}

\begin{lem}\label{exp and alpha}
Let $(\huaL,[\cdot,\cdot,\cdot]_\huaL)$ be a finite-dimensional $3$-Leibniz algebra and $\alpha:\huaL\rightarrow\huaL$ an automorphism. Then for any $x,y\in\huaL$, we have
\begin{eqnarray*}
\alpha\circ\exp(\ad^{R}_{x,y})\circ\alpha^{-1}=\exp(\ad^{R}_{\alpha(x),\alpha(y)}).
\end{eqnarray*}
\end{lem}

\begin{proof}
For any $z\in\huaL$, we have
\begin{eqnarray*}
(\alpha\circ\ad^{R}_{x,y}\circ\alpha^{-1})(z)
=\alpha[\alpha^{-1}(z),x,y]_\huaL
=[z,\alpha(x),\alpha(y)]_\huaL
=\ad^{R}_{\alpha(x),\alpha(y)}(z).
\end{eqnarray*}
Since $\huaL$ is finite-dimensional, we have
\begin{eqnarray*}
\alpha\circ\sum_{k=0}^{N}\frac{(\ad^{R}_{x,y})^{k}}{k!}\circ\alpha^{-1}
=\sum_{k=0}^{N}\frac{\alpha\circ(\ad^{R}_{x,y})^{k}\circ\alpha^{-1}}{k!}
=\sum_{k=0}^{N}\frac{(\alpha\circ\ad^{R}_{x,y}\circ\alpha^{-1})^{k}}{k!}
=\sum_{k=0}^{N}\frac{(\ad^{R}_{\alpha(x),\alpha(y)})^{k}}{k!}.
\end{eqnarray*}
Notice that composition with $\alpha$ is continuous, thus we have
\begin{eqnarray*}
\alpha\circ\exp(\ad^{R}_{x,y})\circ\alpha^{-1}
&=&\alpha\circ\sum_{k=0}^{{+\infty}}\frac{(\ad^{R}_{x,y})^{k}}{k!}\circ\alpha^{-1}
=\alpha\circ\lim_{N\rightarrow {+\infty}}\sum_{k=0}^{N}\frac{(\ad^{R}_{x,y})^{k}}{k!}\circ\alpha^{-1}\\
&=&\lim_{N\rightarrow {+\infty}}\alpha\circ\sum_{k=0}^{N}\frac{(\ad^{R}_{x,y})^{k}}{k!}\circ\alpha^{-1}
=\lim_{N\rightarrow {+\infty}}\sum_{k=0}^{N}\frac{(\ad^{R}_{\alpha(x),\alpha(y)})^{k}}{k!}\\
&=&\sum_{k=0}^{+\infty}\frac{(\ad^{R}_{\alpha(x),\alpha(y)})^{k}}{k!}
=\exp(\ad^{R}_{\alpha(x),\alpha(y)}),
\end{eqnarray*}
which finishes the proof.\end{proof}

\begin{thm}\label{from-3-Lei-to-3-rack-thm}
Let $(\huaL,[\cdot,\cdot,\cdot]_\huaL)$ be a finite-dimensional $3$-Leibniz algebra. Define a ternary map $T:\huaL\times\huaL\times\huaL\rightarrow\huaL$ by
\begin{eqnarray}\label{from-3-Lei-to-3-rack-formula}
T(x,y,z)=\exp(\ad^{R}_{y,z})(x),\,\,\,\forall x,y,z\in\huaL.
\end{eqnarray}
Then $(\huaL,T)$ is a $3$-rack.
\end{thm}

\begin{proof}
For any $x_{1},x_{2},x_{3},y_{1},y_{2}\in\huaL$, we have
\begin{eqnarray*}
 &&T(T(x_{1},x_{2},x_{3}),y_{1},y_{2})\\
 &=&T(\exp(\ad^{R}_{x_2,x_3})(x_1),y_1,y_2)\\
 &=&\exp(\ad^{R}_{y_1,y_2})\exp(\ad^{R}_{x_2,x_3})(x_1)\\
 &=&\exp(\ad^{R}_{y_1,y_2})\exp(\ad^{R}_{x_2,x_3})
    \exp(\ad^{R}_{y_1,y_2})^{-1}\exp(\ad^{R}_{y_1,y_2})(x_1)\\
 &=&\exp(\ad^{R}_{\exp(\ad^{R}_{y_1,y_2})(x_2),\exp(\ad^{R}_{y_1,y_2})(x_3)})
    \exp(\ad^{R}_{y_1,y_2})(x_1)
    \quad\quad\quad\quad{\rm (by~Lemma~\ref{exp and alpha})}\\
 &=&\exp(\ad^{R}_{T(x_2,y_1,y_2),T(x_3,y_1,y_2)})
    T(x_1,y_1,y_2)\\
 &=&T(T(x_1,y_1,y_2),T(x_2,y_1,y_2),T(x_3,y_1,y_2)).
\end{eqnarray*}
Moreover, $T(\cdot,x,y)=\exp(\ad^{R}_{x,y})$ is bijective by Lemma \ref{exp(D)}. Thus, $(\huaL,T)$ is a $3$-rack.
\end{proof}

\begin{ex}
Consider the $3$-Leibniz algebra $(\huaL,[\cdot,\cdot,\cdot]_\huaL)$ with a basis $\{e_1,e_2,e_3\}$ given in {\rm Example \ref{ex-nil-3-lei}}. By {\rm Theorem \ref{from-3-Lei-to-3-rack-thm}}, we obtain a $3$-rack $(\huaL,T)$, where for any $x,y,z\in\{e_1,e_2,e_3\}$,
\begin{align*}
\begin{split}
T(x,y,z)=\left\{
\begin{array}{ll}
\frac{1}{2}e_1+e_2+e_3, & if~x=y=z=e_3,\\
e_1+e_2, & if~x=e_2~and~y=z=e_3,\\
x, & otherwise.
\end{array}
\right.
\end{split}
\end{align*}
\end{ex}

At the end of this section, we sort out the interconnection between Leibniz algebras, $3$-Leibniz algebras, racks and $3$-racks. It is obvious that there exist two ways from $3$-Leibniz algebras to racks.
On the one hand, we can construct a $3$-rack $(\huaL,T)$ by Theorem \ref{from-3-Lei-to-3-rack-thm} and then obtain a rack $(\huaL\times\huaL,\lhd_{T})$ by Lemma \ref{3-rack-to-rack}. On the other hand, there is a Leibniz algebra $(\huaL\otimes\huaL,\{\cdot,\cdot\})$ induced by the $3$-Leibniz algebra $\huaL$ given in Proposition \ref{3-lei-to-lei}, which induces a rack $(\huaL\otimes\huaL,\lhd)$ by Theorem \ref{lei to rack thm}. Moreover, as shown in \cite{AG}, racks naturally give rise to set-theoretical solutions of the Yang-Baxter equation.
\begin{thm}[\cite{AG}]\label{rack-set-YBO}
Let $X$ be a set and $\lhd:X\times X\to X$ a binary operator. Define
$$R:X\times X\to X\times X,\,\,\,\,\,\,R(x,y)=(y,x\lhd y),$$
then $R$ is a set-theoretical solution  of the Yang-Baxter equation if and only if $(X,\lhd)$ is a rack.\end{thm}

Define a map $\phi$ from $\huaL\times\huaL$ to $\huaL\otimes\huaL$ as follows:
\begin{eqnarray*}
\phi(x_1,x_2)=x_1\otimes x_2,\,\,\,\forall (x_1,x_2)\in\huaL\times\huaL.
\end{eqnarray*}

\begin{thm}\label{two-ways-from-3-Lei-to-rack}
Let $(\huaL,[\cdot,\cdot,\cdot]_\huaL)$ be a  $3$-Leibniz algebra. Then $\phi:\huaL\times\huaL\to\huaL\otimes\huaL$ is a rack morphism from $(\huaL\times\huaL,\lhd_{T})$ to $(\huaL\otimes\huaL,\lhd)$.
Moreover, for the solutions of the Yang-Baxter equation $R^{\lhd_{T}}$ and $R^{\lhd}$, which are induced by racks $(\huaL\times\huaL,\lhd_{T})$ and $(\huaL\otimes\huaL,\lhd)$ respectively, $\phi$ is a homomorphism from $R^{\lhd_{T}}$ to $R^{\lhd}$, that is,  $(\phi\times\phi)R^{\lhd_T}=R^{\lhd}(\phi\times\phi)$ (see Diagram \eqref{diagram:sec2}).
 \end{thm}

\begin{proof}
For any $(x_1,x_2),(y_1,y_2)\in\huaL\times\huaL$, we have
\begin{eqnarray*}
\phi\Big((x_1,x_2)\lhd_{T}(y_1,y_2)\Big)
&\stackrel{\eqref{3-rack-to-rack-formula}}{=}&\phi(T(x_1,y_1,y_2),T(x_2,y_1,y_2))\\
&
\stackrel{\eqref{from-3-Lei-to-3-rack-formula}}{=}&\phi(\exp(\ad_{y_1,y_2}^{R})(x_1),\exp(\ad_{y_1,y_2}^{R})(x_2))\\
&=&\sum_{i=0}^{+\infty}\frac{(\ad_{y_1,y_2}^{R})^{i}}{i!}(x_1)
\otimes
\sum_{j=0}^{+\infty}\frac{(\ad_{y_1,y_2}^{R})^{j}}{j!}(x_2)\\
&=&\sum_{k=0}^{+\infty}\sum_{\substack{{i+j}=k\\i,j\geqslant 0}}
\frac{(\ad_{y_1,y_2}^{R})^{i}}{i!}(x_1)\otimes
\frac{(\ad_{y_1,y_2}^{R})^{j}}{j!}(x_2).
\end{eqnarray*}

The formula \eqref{fundamental-element-formula} shows that
\begin{eqnarray*}
\ad^{R}_{y_1\otimes y_2}=
\ad_{y_1,y_2}^{R}\otimes{\Id}+{\Id}\otimes\ad_{y_1,y_2}^{R},
\end{eqnarray*}
where $\ad^{R}_{y_1\otimes y_2}:\huaL\otimes\huaL\rightarrow\huaL\otimes\huaL$ is defined by $\ad^{R}_{y_1\otimes y_2}(x_1\otimes x_2)
=\{x_1\otimes x_2,y_1\otimes y_2\}$.
Thus we have
\begin{eqnarray*}
(\ad^{R}_{y_1\otimes y_2})^k
&=&(\ad_{y_1,y_2}^{R}\otimes{\Id}+{\Id}\otimes\ad_{y_1,y_2}^{R})^k\\
&=&\sum_{\substack{i+j=k\\i,j\geqslant 0}}
\frac{k!}{i!j!}(\ad_{y_1,y_2}^{R}\otimes{\Id})^{i}({\Id}\otimes\ad_{y_1,y_2}^{R})^{j}\\
&=&\sum_{\substack{i+j=k\\i,j\geqslant 0}}
\frac{k!}{i!j!}(\ad_{y_1,y_2}^{R})^{i}\otimes(\ad_{y_1,y_2}^{R})^{j},
\end{eqnarray*}
which shows that
\begin{eqnarray*}
\phi(x_1,x_2)\lhd \phi(y_1,y_2)
&=&(x_1\otimes x_2)\lhd(y_1\otimes y_2)\\
&=&\exp(\ad_{y_1\otimes y_2}^{R})(x_1\otimes x_2)
    \quad\quad\quad\quad{\rm (by~Theorem~\ref{lei to rack thm})}\\
&=&\sum_{k=0}^{+\infty}\frac{(\ad^{R}_{y_1\otimes y_2})^k}{k!}(x_1\otimes x_2)\\
&=&\Big(\sum_{k=0}^{+\infty}\sum_{\substack{i+j=k\\i,j\geqslant 0}}
\frac{1}{i!j!}(\ad_{y_1,y_2}^{R})^{i}\otimes(\ad_{y_1,y_2}^{R})^{j}\Big)
(x_1\otimes x_2)\\
&=&\sum_{k=0}^{+\infty}\sum_{\substack{{i+j}=k\\i,j\geqslant 0}}
\frac{(\ad_{y_1,y_2}^{R})^{i}}{i!}(x_1)\otimes
\frac{(\ad_{y_1,y_2}^{R})^{j}}{j!}(x_2).
\end{eqnarray*}
Therefore, we obtain $\phi\Big((x_1,x_2)\lhd_{T}(y_1,y_2)\Big)=\phi(x_1,x_2)\lhd \phi(y_1,y_2)$, i.e. $\phi$ is a rack morphism from $(\huaL\times\huaL,\lhd_{T})$ to $(\huaL\otimes\huaL,\lhd)$.

By Theorem \ref{rack-set-YBO}, set-theoretical solutions of the Yang-Baxter equation induced by racks $(\huaL\times\huaL,\lhd_{T})$ and $(\huaL\otimes\huaL,\lhd)$ respectively are as follows:
\begin{eqnarray*}
R^{\lhd_T}:(\huaL\times\huaL)\times(\huaL\times\huaL)&\longrightarrow&(\huaL\times\huaL)\times(\huaL\times\huaL)\\
\Big((x_1,x_2),(y_1,y_2)\Big)&\longmapsto&\Big((y_1,y_2),(x_1,x_2)\lhd_T(y_1,y_2)\Big),\\
R^{\lhd}:(\huaL\otimes\huaL)\times(\huaL\otimes\huaL)&\longrightarrow&(\huaL\otimes\huaL)\times(\huaL\otimes\huaL)\\
\Big((x_1\otimes x_2),(y_1\otimes y_2)\Big)&\longmapsto&\Big((y_1\otimes y_2),(x_1,x_2)\lhd(y_1\otimes y_2)\Big).
\end{eqnarray*}
For all $(x_1,x_2),(y_1,y_2)\in \huaL\times\huaL$, we have
\begin{eqnarray*}
(\phi\times\phi)R^{\lhd_T}\Big((x_1,x_2),(y_1,y_2)\Big)
&=&(\phi\times\phi)\Big((y_1,y_2),(x_1,x_2)\lhd_T(y_1,y_2)\Big)\\
&=&\Big(\phi(y_1,y_2),\phi\big((x_1,x_2)\lhd_T(y_1,y_2)\big)\Big)\\
&=&\Big(\phi(y_1,y_2),\phi(x_1,x_2)\lhd\phi(y_1,y_2)\Big)\\
&=&\Big(y_1\otimes y_2,(x_1\otimes x_2)\lhd(y_1\otimes y_2)\Big)\\
&=&R^{\lhd}(x_1\otimes x_2,y_1\otimes y_2)\\
&=&R^{\lhd}(\phi\times\phi)\Big((x_1,x_2),(y_1,y_2)\Big),
\end{eqnarray*}
which implies that $\phi$ is a homomorphism from $R^{\lhd_{T}}$ to $R^{\lhd}$.
\end{proof}

\section{A passage from trilinear racks to linear racks}\label{sec-lin-rack-tri-rack}

In this section, we show that a trilinear rack $\huaC$ naturally gives rise to a linear rack structure on $\huaC\otimes \huaC$. Furthermore, the solution of the Yang-Baxter equation induced by the linear rack  $\huaC\otimes \huaC$ is consistent with the solution induced by the trilinear rack $\huaC$ given in \cite{Abramov}.
This alternatively give an explanation of the construction in \cite{Abramov}.

\subsection{Linear racks and the Yang-Baxter equation}
\begin{defi}
Let $\huaC$ be a vector space and $\triangle:\huaC\rightarrow\huaC\otimes\huaC$, $\varepsilon:\huaC\rightarrow\K$ two linear maps.
If $\triangle$ and $\varepsilon$ satisfy the following commutative diagrams:
\vspace{-0.5cm}
\begin{center}
\begin{equation*}
\xymatrix@C=6ex@R=6ex{
  \txt{$\huaC$}
  \ar[r]^-{\triangle} \ar[d]_-{\triangle}
  &\txt{$\huaC\otimes\huaC$}
  \ar[d]^-{\triangle\otimes{\Id}}
  & &\txt{$\huaC$}
  \ar[d]^-{\triangle} \ar[dl]_-{\cong} \ar[dr]^-{\cong}
  &\\
  \txt{$\huaC\otimes\huaC$}
  \ar[r]_-{{\Id}\otimes\triangle}
  &\txt{$\huaC\otimes\huaC\otimes\huaC$}
  &\txt{$\K\otimes\huaC$}
  &\txt{$\huaC\otimes\huaC$}
  \ar[r]_-{{\Id}\otimes\varepsilon} \ar[l]^-{\varepsilon\otimes{\Id}}
  &\txt{$\huaC\otimes\K$}
}
\end{equation*}
\end{center}
then $(\huaC,\triangle,\varepsilon)$ is called a {\bf coassociative counital coalgebra}.

Moreover, if $\tau\circ\triangle=\triangle$, where the linear map $\tau:\huaC\otimes\huaC\rightarrow\huaC\otimes\huaC$ is defined by $\tau(x\otimes y)=y\otimes x$ for any $x\otimes y\in\huaC\otimes\huaC$, then the coalgebra $(\huaC,\triangle,\varepsilon)$ is called {\bf cocommutative}.
\end{defi}

\begin{pro}\label{ex-of-coalg}
Let $(\huaC,\triangle,\varepsilon)$ and $(\huaC',\triangle',\varepsilon')$ be two coassociative counital coalgebras. Then $(\huaC\otimes\huaC',\triangle_{\huaC\otimes\huaC'},\varepsilon_{\huaC\otimes\huaC'})$ is also a coassociative counital coalgebra, where for any $x\otimes y\in\huaC\otimes\huaC'$,
\begin{eqnarray}\label{coproduct-on-tensor}
\triangle_{\huaC\otimes\huaC'}(x\otimes y)=(x_{(1)}\otimes y_{(1)})\otimes(x_{(2)}\otimes y_{(2)}),\,\,\,\,\,\,
\varepsilon_{\huaC\otimes\huaC'}(x\otimes y)=\varepsilon_{\huaC}(x)\varepsilon_{\huaC'}(y).
\end{eqnarray}
Here we use Sweedler notation $\triangle(x)=x_{(1)}\otimes x_{(2)}$ for the coproduct $\triangle$.
\end{pro}

\begin{defi}
Let $(\huaC,\triangle,\varepsilon)$ and $(\huaC',\triangle',\varepsilon')$ be two coassociative counital coalgebras. A linear map $f:\huaC\rightarrow\huaC'$ is called a {\bf coalgebra morphism} if
$\triangle'\circ f=(f\otimes f)\circ\triangle$ and $\varepsilon'\circ f=\varepsilon$.
\end{defi}

\begin{defi}
Let $(\huaC,\triangle,\varepsilon)$ be a coassociative counital coalgebra and $\lhd:\huaC\otimes\huaC\rightarrow\huaC$ a coalgebra morphism. If $\lhd$ satisfies the general self-distributive law:
\begin{eqnarray}\label{gen-self-dis}
(u\lhd v)\lhd w&=&(u\lhd w_{(1)})\lhd(v\lhd w_{(2)}),\,\,\,
\forall u,v,w\in\huaC,\label{quantum shelf def1}
\end{eqnarray}
then $(\huaC,\lhd)$ is called a {\bf linear shelf}.

Moreover, if there exists a coalgebra morphism $\widetilde{\lhd}:\huaC\otimes\huaC\rightarrow\huaC$ such that $(\huaC,\widetilde{\lhd})$ is a linear shelf and the following condition holds:
\begin{eqnarray}\label{twist inverse}
(u\lhd v_{(2)})~\widetilde{\lhd}~v_{(1)}
=\varepsilon(v)u
=(u~\widetilde{\lhd}~v_{(2)})\lhd v_{(1)},\,\,\,\forall u,v\in\huaC,
\end{eqnarray}
then the linear shelf $(\huaC,\lhd)$ is called a {\bf linear rack}.
\end{defi}

\begin{defi}
Let $(\huaC,\lhd)$ and $(\huaC',{\lhd}')$ be two linear racks. A coalgebra morphism $f:\huaC\to\huaC'$ is called a {\bf linear rack morphism} if $f(x\lhd y)=f(x)\lhd' f(y)$ for all $x,y\in\huaC$.
\end{defi}

\begin{ex}\label{rack-to-lin-rack}\cite{ABRW}
Let $(X,\lhd)$ be a rack. Consider the vector space $\K[X]$ which has the elements of $X$ as a basis. $\K[X]$ equipped with the following coproduct $\triangle:\K[X]\to\K[X]\otimes\K[X]$ and counit $\varepsilon:\K[X]\to\K$ is a cocommutative coassociative counital coalgebra:
\begin{eqnarray}\label{set-to-coalg}
\triangle(x)=x\otimes x,\,\,\,\varepsilon(x)=1_{\K},\,\,\,\forall x\in X.
\end{eqnarray}
Extending $\lhd$ linearly to the coalgebra $\K[X]$, one obtain a linear rack $(\K[X],\lhd^{\K[X]})$.
\end{ex}

Let $(\huaE,[\cdot,\cdot]_\huaE)$ be a Leibniz algebra. Consider the vector space $\K\oplus\huaE$. Define two linear maps $\triangle:\K\oplus\huaE\to(\K\oplus\huaE)\otimes(\K\oplus\huaE)$ and $\varepsilon:\K\oplus\huaE\to\K$ as follows:
\begin{eqnarray*}
&&\triangle(1_\K)=1_\K\otimes1_\K,\,\,\,\varepsilon(1_\K)=1_\K;\,\,\,
\triangle(x)=x\otimes1_\K+1_\K\otimes x,\,\,\,\varepsilon(x)=0,\,\,\,\forall x\in\huaE.
\end{eqnarray*}
It is obvious that $(\K\oplus\huaE,\triangle,\varepsilon)$ is a cocommutative coassociative counital coalgebra. Define coalgebra morphisms $\lhd,\widetilde{\lhd}:\Big(\K\oplus\huaE\Big)\otimes\Big(\K\oplus\huaE\Big)\rightarrow\K\oplus\huaE$
as follows:
\begin{eqnarray*}
&&1_\K\lhd1_\K=1_\K,\,\,\,\,1_\K\lhd x=0,\,\,\,x\lhd1_\K=x,\,\,\,x\lhd y=[x,y]_\huaE;\\
&&1_\K~\widetilde{\lhd}~1_\K=1_\K,\,\,\,1_\K~\widetilde{\lhd}~x=0,\,\,\,x~\widetilde{\lhd}~1_\K=x,\,\,\,x~\widetilde{\lhd}~y=-[x,y]_\huaE.
\end{eqnarray*}

\begin{pro}\label{cen-Lei-to-lin-rack}\cite{Lebed3}
With the above notations, $(\K\oplus\huaE,\lhd)$ is a linear rack.
\end{pro}

\begin{pro}\cite{Lebed3}\label{lin-rack-to-solution}
Let $(\huaC,\lhd)$ be a linear shelf. If the coalgebra $\huaC$ is cocommutative, then $R^{\lhd}:\huaC\otimes\huaC\rightarrow\huaC\otimes\huaC$ defined by
\begin{eqnarray}\label{quantum rack to solution formula}
R^{\lhd}(u\otimes v)=v_{(1)}\otimes (u\lhd v_{(2)}),\,\,\,\forall u\otimes v\in\huaC\otimes\huaC,
\end{eqnarray}
satisfies the equation \eqref{YBE}.

Moreover, if $(\huaC,\lhd)$ is a linear rack, then $R^{\lhd}$ is a solution of the Yang-Baxter equation with the inverse
\begin{eqnarray}
(R^{\lhd})^{ -1}(u\otimes v)=(v~\widetilde{\lhd}~u_{(2)})\otimes u_{(1)},\,\,\,\forall u\otimes v\in\huaC\otimes\huaC.
\end{eqnarray}
\end{pro}

Recall that a {\bf central Leibniz algebra} is a Leibniz algebra  $(\huaE,[\cdot,\cdot]_\huaE)$ with a central element, i.e. an element ${\bf 1}\in\huaE$ satisfies
$[{\bf 1},x]_\huaE=0=[x,{\bf 1}]_\huaE$ for all $x\in\huaE$.
 In \cite{Lebed1}, Lebed showed that a central Leibniz algebra gives rise to a solution of the Yang-Baxter equation as follows:
\begin{thm}\cite{Lebed1}\label{lei to solution}
Let $(\huaE,[\cdot,\cdot]_\huaE,{\bf 1})$ be a central Leibniz algebra. Then $R^{Lei}:\huaE\otimes\huaE\rightarrow\huaE\otimes\huaE$ defined by
\begin{eqnarray}\label{lei to solution formula}
R^{Lei}(x\otimes y)=y\otimes x+{\bf 1}\otimes[x,y]_\huaE,\,\,\,\forall x,y\in\huaE,
\end{eqnarray}
is a solution of the Yang-Baxter equation.
\end{thm}
 Note that $\K\oplus\huaE$ showed in Proposition \ref{cen-Lei-to-lin-rack} is actually a central Leibniz algebra, where $\K$ is the set of central elements of $\K\oplus\huaE$. Lebed showed in \cite{Lebed3} that starting from the central Leibniz algebra $\K\oplus\huaE$, whether use Proposition \ref{cen-Lei-to-lin-rack} to get the linear rack $(\K\oplus\huaE,\lhd)$, then use Proposition \ref{lin-rack-to-solution} to get the solution of the Yang-Baxter equation $R^\lhd_{\K\oplus\huaE}$ on $\K\oplus\huaE$,
 or directly use Theorem \ref{lei to solution} to get the solution of the Yang-Baxter equation $R^{Lei}_{\K\oplus\huaE}$ on $\K\oplus\huaE$, both ways give the same result:
\begin{eqnarray*}\label{Lei-cen-ext-to-solution}
R^{Lei}_{\K\oplus\huaE}((a,x)\otimes(b,y))
&=&(b,y)\otimes(a,x)+(1_\K,0)\otimes(0,[x,y]_\huaE)\\
&=&R^{\lhd}_{\K\oplus\huaE}((a,x)\otimes(b,y)),\,\,\,
\forall (a,x),(b,y)\in\K\oplus\huaE.
\end{eqnarray*}

\subsection{Trilinear racks and the Yang-Baxter equation}
\emptycomment{
Let $(\huaC,\triangle,\varepsilon)$ be a coassociative counital coalgebra. By Example \ref{ex-of-coalg}, $\huaC\otimes\huaC\otimes\huaC$ is also a coassociative counital coalgebra. Using the Sweedler notation, the coproduct $\triangle_3:\huaC\otimes\huaC\otimes\huaC\to(\huaC\otimes\huaC\otimes\huaC)\otimes(\huaC\otimes\huaC\otimes\huaC)$ is defined by
$$\triangle_3(x\otimes y\otimes z)=(x_{(1)}\otimes y_{(1)}\otimes z_{(1)})\otimes(x_{(2)}\otimes y_{(2)}\otimes z_{(2)}).$$}
\begin{defi}
Let $(\huaC,\triangle,\varepsilon)$ be a coassociative counital coalgebra and $T:\huaC\otimes\huaC\otimes\huaC\rightarrow\huaC$ a coalgebra morphism. If for any $x,y,z,u,v\in\huaC$,
\begin{eqnarray}
T(T(x,y,z),u,v)=T(T(x,u_{(1)},v_{(1)}),T(y,u_{(2)},v_{(2)}),T(z,u_{(3)},v_{(3)})),\label{TSD-formula}
\end{eqnarray}
where $(\triangle\otimes{\Id})\circ\triangle(u)=u_{(1)}\otimes u_{(2)}\otimes u_{(3)}$ for any $u\in\huaC$, then $(\huaC,T)$ is called a {\bf trilinear shelf}.

Moreover, if there exists a coalgebra morphism $\widetilde{T}:\huaC\otimes\huaC\otimes\huaC\rightarrow\huaC$ such that $(\huaC,\widetilde{T})$ is a trilinear shelf and the following condition holds:
\begin{eqnarray}
\quad\widetilde{T}(T(x,y_{(2)},z_{(2)}),z_{(1)},y_{(1)})=\varepsilon(y)\varepsilon(z)x
=T(\widetilde{T}(x,y_{(2)},z_{(2)}),z_{(1)},y_{(1)}),\,\,\,\forall x,y,z\in\huaC,\label{reversibility}
\end{eqnarray}
then the trilinear shelf $(\huaC,T)$ is called a {\bf trilinear rack}.
\end{defi}

\begin{ex}\label{3-rack-to-tri-rack}
Let $(X,T)$ be a $3$-rack. Consider the cocommutative coassociative counital coalgebra $(\K[X],\triangle,\varepsilon)$, where $\triangle$ and $\varepsilon$ are given by \eqref{set-to-coalg}.
Extending $T$ linearly to the coalgebra $\K[X]$, we obtain a trilinear rack $(\K[X],T^{\K[X]})$.
\end{ex}

Abramov and Zappala proved that given a trilinear rack $(\huaC,T)$, if $\huaC$ is cocommutative, then one can construct a solution of the Yang-Baxter equation on $\huaC\otimes\huaC$ as follows:
\begin{thm}\cite{Abramov}\label{TSD-to-solution}
Let $(\huaC,T)$ be a trilinear rack where $\huaC$ is cocommutative. Define $R^T_{\huaC\otimes\huaC}:(\huaC\otimes\huaC)\otimes(\huaC\otimes\huaC)\rightarrow(\huaC\otimes\huaC)\otimes(\huaC\otimes\huaC)$
by:
\begin{eqnarray}
R^{T}_{\huaC\otimes\huaC}\Big((u\otimes v)\otimes(m\otimes n)\Big)
=(m_{(1)}\otimes n_{(1)})\otimes\Big(T(u,m_{(2)},n_{(2)})\otimes T(v,m_{(3)},n_{(3)})\Big).\label{TSD-to-solution-formula}
\end{eqnarray}
Then   $R^T_{\huaC\otimes\huaC}$ is a solution of the Yang-Baxter equation.
\end{thm}

We will show that the above solution induced by a trilinear rack can be explained in terms of a linear rack actually.
\begin{thm}\label{TSD-to-quan-rack}
Let $(\huaC,T)$ be a trilinear shelf where $\huaC$ is cocommutative. Then $\huaC\otimes\huaC$ with the linear map $\lhd_T:(\huaC\otimes\huaC)\otimes(\huaC\otimes\huaC)\to\huaC\otimes\huaC$ defined as follows is a linear shelf:
\begin{eqnarray*}
(u\otimes v)\lhd_T(m\otimes n)=T(u,m_{(1)},n_{(1)})\otimes T(v,m_{(2)},n_{(2)}),\,\,\,\forall u\otimes v,m\otimes n\in\huaC\otimes\huaC.
\end{eqnarray*}
In addition, if $(\huaC,T)$ is a trilinear rack, then $(\huaC\otimes\huaC,\lhd_T)$ is a linear rack.
Moreover, the solution of the Yang-Baxter equation $R_{\huaC\otimes\huaC}^T$ induced by a trilinear rack $(\huaC,T)$ given in \cite{Abramov} is consistent with the solution $R_{\huaC\otimes\huaC}^{\lhd_T}$ induced by the linear rack $(\huaC\otimes\huaC,\lhd_T)$, as explained in Diagram \eqref{diagram:sec3}.
\end{thm}

\begin{proof}
For any $u\otimes v,m\otimes n,p\otimes q\in\huaC\otimes\huaC$, we have
\begin{eqnarray*}
&&\Big((u\otimes v)\lhd_T (m\otimes n)\Big)\lhd_T (p\otimes q)\\
&=&\Big(T(u,m_{(1)},n_{(1)})\otimes T(v,m_{(2)},n_{(2)})\Big)\lhd_T (p\otimes q)\\
&=&T\Big(T(u,m_{(1)},n_{(1)}),p_{(1)},q_{(1)}\Big)\otimes T\Big(T(v,m_{(2)},n_{(2)}),p_{(2)},q_{(2)}\Big)\\
&\stackrel{\eqref{TSD-formula}}{=}&
T\Big(T(u,p_{(1)(1)},q_{(1)(1)}),T(m_{(1)},p_{(1)(2)},q_{(1)(2)}),T(n_{(1)},p_{(1)(3)},q_{(1)(3)})\Big)\\
&&\otimes T\Big(T(v,p_{(2)(1)},q_{(2)(1)}),T(m_{(2)},p_{(2)(2)},q_{(2)(2)}),T(n_{(2)},p_{(2)(3)},q_{(2)(3)})\Big)\\
&=&T\Big(T(u,p_{(1)},q_{(1)}),T(m_{(1)},p_{(2)},q_{(2)}),T(n_{(1)},p_{(3)},q_{(3)})\Big)\\
&&\otimes T\Big(T(v,p_{(4)},q_{(4)}),T(m_{(2)},p_{(5)},q_{(5)}),T(n_{(2)},p_{(6)},q_{(6)})\Big),\\
&&\Big((u\otimes v)\lhd_T (p\otimes q)_{(1)}\Big)\lhd_T \Big((m\otimes n)\lhd_T (p\otimes q)_{(2)}\Big)\\
&\stackrel{\eqref{coproduct-on-tensor}}{=}&\Big((u\otimes v)\lhd_T (p_{(1)}\otimes q_{(1)})\Big)\lhd_T \Big((m\otimes n)\lhd_T (p_{(2)}\otimes q_{(2)})\Big)\\
&=&\Big(T(u,p_{(1)(1)},q_{(1)(1)})\otimes T(v,p_{(1)(2)},q_{(1)(2)})\Big)\lhd_T
\Big(T(m,p_{(2)(1)},q_{(2)(1)})\otimes T(n,p_{(2)(2)},q_{(2)(2)})\Big)\\
&=&T\Big(T(u,p_{(1)(1)},q_{(1)(1)}),T(m,p_{(2)(1)},q_{(2)(1)})_{(1)},T(n,p_{(2)(2)},q_{(2)(2)})_{(1)}\Big)\\
&&\otimes T\Big(T(v,p_{(1)(2)},q_{(1)(2)}),T(m,p_{(2)(1)},q_{(2)(1)})_{(2)},T(n,p_{(2)(2)},q_{(2)(2)})_{(2)}\Big)\\
&=&T\Big(T(u,p_{(1)(1)},q_{(1)(1)}),T(m_{(1)},p_{(2)(1)(1)},q_{(2)(1)(1)}),T(n_{(1)},p_{(2)(2)(1)},q_{(2)(2)(1)})\Big)\\
&&\otimes T\Big(T(v,p_{(1)(2)},q_{(1)(2)}),T(m_{(2)},p_{(2)(1)(2)},q_{(2)(1)(2)}),T(n_{(2)},p_{(2)(2)(2)},q_{(2)(2)(2)})\Big)\\
&=&T\Big(T(u,p_{(1)},q_{(1)}),T(m_{(1)},p_{(3)},q_{(3)}),T(n_{(1)},p_{(5)},q_{(5)})\Big)\\
&&\otimes T\Big(T(v,p_{(2)},q_{(2)}),T(m_{(2)},p_{(4)},q_{(4)}),T(n_{(2)},p_{(6)},q_{(6)})\Big).
\end{eqnarray*}
Since $\huaC$ is cocommutative, then we have
$$\Big((u\otimes v)\lhd_T (m\otimes n)\Big)\lhd_T (p\otimes q)=\Big((u\otimes v)\lhd_T (p\otimes q)_{(1)}\Big)\lhd_T \Big((m\otimes n)\lhd_T (p\otimes q)_{(2)}\Big),$$
which implies that $\lhd_T$ satisfies the general self-distributive law \eqref{gen-self-dis}.
Since $T$ is a coalgebra morphism and $\huaC$ is cocommutative, it is easy to obtain that
$$ \Big((u\otimes v)\lhd_T (m\otimes n)\Big)_{(1)}\otimes\Big((u\otimes v)\lhd_T (m\otimes n)\Big)_{(2)}
=\Big((u\otimes v)_{(1)}\lhd_T (m\otimes n)_{(1)}\Big)\otimes\Big((u\otimes v)_{(2)}\lhd_T (m\otimes n)_{(2)}\Big),$$
which implies that $\lhd_T$ is a coalgebra morphism from the coalgebra $(\huaC\otimes\huaC)\otimes(\huaC\otimes\huaC)$ to the coalgebra $\huaC\otimes\huaC$.
\emptycomment{
\begin{eqnarray*}
&&\Big((u\otimes v)\lhd_T (m\otimes n)\Big)_{(1)}\otimes\Big((u\otimes v)\lhd_T (m\otimes n)\Big)_{(2)}\\
&=&\Big(T(u,m_{(1)},n_{(1)})\otimes T(v,m_{(2)},n_{(2)})\Big)_{(1)}\otimes\Big(T(u,m_{(1)},n_{(1)})\otimes T(v,m_{(2)},n_{(2)})\Big)_{(2)}\\
&=&\Big(T(u,m_{(1)},n_{(1)})_{(1)}\otimes T(v,m_{(2)},n_{(2)})_{(1)}\Big)\otimes\Big(T(u,m_{(1)},n_{(1)})_{(2)}\otimes T(v,m_{(2)},n_{(2)})_{(2)}\Big)\\
&=&\Big(T(u_{(1)},m_{(1)(1)},n_{(1)(1)})\otimes T(v_{(1)},m_{(2)(1)},n_{(2)(1)})\Big)\otimes\Big(T(u_{(2)},m_{(1)(2)},n_{(1)(2)})\otimes T(v_{(2)},m_{(2)(2)},n_{(2)(2)})\Big)\\
&=&\Big(T(u_{(1)},m_{(1)},n_{(1)})\otimes T(v_{(1)},m_{(3)},n_{(3)})\Big)\otimes\Big(T(u_{(2)},m_{(2)},n_{(2)})\otimes T(v_{(2)},m_{(4)},n_{(4)})\Big)\\
&&\Big((u\otimes v)_{(1)}\lhd_T (m\otimes n)_{(1)}\Big)\otimes\Big((u\otimes v)_{(2)}\lhd_T (m\otimes n)_{(2)}\Big)\\
&=&\Big((u_{(1)}\otimes v_{(1)})\lhd_T (m_{(1)}\otimes n_{(1)})\Big)\otimes\Big((u_{(2)}\otimes v_{(2)})\lhd_T (m_{(2)}\otimes n_{(2)})\Big)\\
&=&\Big(T(u_{(1)},m_{(1)(1)},n_{(1)(1)})\otimes T(v_{(1)},m_{(1)(2)},n_{(1)(2)})\Big)\otimes\Big(T(u_{(2)},m_{(2)(1)},n_{(2)(1)})\otimes T(v_{(2)},m_{(2)(2)},n_{(2)(2)})\Big)\\
&=&\Big(T(u_{(1)},m_{(1)},n_{(1)})\otimes T(v_{(1)},m_{(2)},n_{(2)})\Big)\otimes\Big(T(u_{(2)},m_{(3)},n_{(3)})\otimes T(v_{(2)},m_{(4)},n_{(4)})\Big)
\end{eqnarray*}
}
Therefore, $(\huaC,\lhd_T)$ is a linear shelf.

If there exists a coalgebra morphism $\widetilde{T}:\huaC\otimes\huaC\otimes\huaC\rightarrow\huaC$ satisfying \eqref{reversibility}, then we define $\widetilde{\lhd_T}:(\huaC\otimes\huaC)\otimes(\huaC\otimes\huaC)\rightarrow\huaC\otimes\huaC$ as follows:
\begin{eqnarray*}
(u\otimes v)\widetilde{\lhd_T}(m\otimes n)
=\widetilde{T}(u,n_{(1)},m_{(1)})\otimes \widetilde{T}(v,n_{(2)},m_{(2)}),\,\,\,\forall u\otimes v,m\otimes n\in\huaC\otimes\huaC.
\end{eqnarray*}
By direct computation, we have
\begin{eqnarray*}
&&\Big((u\otimes v)\lhd_T (m\otimes n)_{(2)}\Big)~\widetilde{\lhd_T}~(m\otimes n)_{(1)}\\
&=&\Big((u\otimes v)\lhd_T (m_{(2)}\otimes n_{(2)})\Big)~\widetilde{\lhd_T}~(m_{(1)}\otimes n_{(1)})\\
&=&\Big(T(u,m_{(2)(1)},n_{(2)(1)})\otimes T(v,m_{(2)(2)},n_{(2)(2)})\Big)~\widetilde{\lhd_T}~(m_{(1)}\otimes n_{(1)})\\
&=&\widetilde{T}(T(u,m_{(2)(1)},n_{(2)(1)}),n_{(1)(1)},m_{(1)(1)})\otimes \widetilde{T}(T(v,m_{(2)(2)},n_{(2)(2)}),n_{(1)(2)},m_{(1)(2)})\\
&=&\widetilde{T}(T(u,m_{(1)(2)},n_{(1)(2)}),n_{(1)(1)},m_{(1)(1)})\otimes \widetilde{T}(T(v,m_{(2)(2)},n_{(2)(2)}),n_{(2)(1)},m_{(2)(1)})\\
&\stackrel{\eqref{reversibility}}{=}&\varepsilon(m_{(1)})\varepsilon(n_{(1)})u\otimes\varepsilon(m_{(2)})\varepsilon(n_{(1)})v\\
&=&\varepsilon(m)\varepsilon(n)u\otimes v
\stackrel{\eqref{coproduct-on-tensor}}{=}\varepsilon(m\otimes n)u\otimes v,\\
&&\Big((u\otimes v)~\widetilde{\lhd_T}~(m\otimes n)_{(2)}\Big)\lhd_T (m\otimes n)_{(1)}\\
&=&\Big((u\otimes v)~\widetilde{\lhd_T}~(m_{(2)}\otimes n_{(2)})\Big)\lhd_T (m_{(1)}\otimes n_{(1)})\\
&=&\Big(\widetilde{T}(u,n_{(2)(1)},m_{(2)(1)})\otimes \widetilde{T}(v,n_{(2)(2)},m_{(2)(2)})\Big)\lhd_T (m_{(1)}\otimes n_{(1)})\\
&=&T(\widetilde{T}(u,n_{(2)(1)},m_{(2)(1)}),m_{(1)(1)},n_{(1)(1)})\otimes T(\widetilde{T}(v,n_{(2)(2)},m_{(2)(2)}),m_{(1)(2)},n_{(1)(2)})\\
&=&T(\widetilde{T}(u,n_{(1)(2)},m_{(1)(2)}),m_{(1)(1)},n_{(1)(1)})\otimes T(\widetilde{T}(v,n_{(2)(2)},m_{(2)(2)}),m_{(2)(1)},n_{(2)(1)})\\
&\stackrel{\eqref{reversibility}}{=}&\varepsilon(n_{(1)})\varepsilon(m_{(1)})u\otimes\varepsilon(n_{(1)})\varepsilon(m_{(2)})v\\
&=&\varepsilon(m)\varepsilon(n)u\otimes v=\varepsilon(m\otimes n)u\otimes v,
\end{eqnarray*}
which means that
$$\Big((u\otimes v)\lhd_T (m\otimes n)_{(2)}\Big)~\widetilde{\lhd_T}~(m\otimes n)_{(1)}=\varepsilon(m\otimes n)u\otimes v=\Big((u\otimes v)~\widetilde{\lhd_T}~(m\otimes n)_{(2)}\Big)\lhd_T (m\otimes n)_{(1)}.$$
Therefore, $(\huaC\otimes\huaC,\lhd_T)$ is a linear rack.

By Proposition \ref{lin-rack-to-solution}, we obtain a solution of the Yang-Baxter equation $R^{\lhd_T}_{\huaC\otimes\huaC}:(\huaC\otimes\huaC)\otimes(\huaC\otimes\huaC)\rightarrow(\huaC\otimes\huaC)\otimes(\huaC\otimes\huaC)$ by the linear rack $(\huaC\otimes\huaC,\lhd_T)$ as follows:
\begin{eqnarray*}
R^{\lhd_T}_{\huaC\otimes\huaC}\Big((u\otimes v)\otimes(m\otimes n)\Big)
&=&(m\otimes n)_{(1)}\otimes\Big((u\otimes v)\lhd_T(m\otimes n)_{(2)}\Big)\\
&=&(m_{(1)}\otimes n_{(1)})\otimes\Big(T(u,m_{(2)},n_{(2)})\otimes T(v,m_{(3)},n_{(3)})\Big),
\end{eqnarray*}
which is the same as the solution $R^T_{\huaC\otimes\huaC}$ (see \eqref{TSD-to-solution-formula}) induced by the trilinear rack $(\huaC,T)$.
\end{proof}
\emptycomment{
Lemma \ref{3-rack-to-rack} and Theorem \ref{TSD-to-quan-rack} respectively show a way from $3$-racks to racks and a way from trilinear racks to linear rack. Example \ref{rack-to-lin-rack} and Example \ref{3-rack-to-tri-rack} respectively show  a way from racks to linear racks and a way from $3$-racks to trilinear racks.}
Now we explore the relationship between racks, $3$-racks, linear racks and trilinear racks. Starting from a $3$-rack $(X,T)$, on the one hand, there is a rack $(X\times X,\lhd_T)$ by Lemma \ref{3-rack-to-rack}, then we obtain a linear rack $(\K[X\times X],\lhd_T^{\K[X\times X]})$ by Example \ref{rack-to-lin-rack}. On the other hand, we construct a trilinear rack $(\K[X],T^{\K[X]})$ by Example \ref{3-rack-to-tri-rack}, which induces a linear rack $(\K[X]\otimes\K[X],\lhd_{T^{\K[X]}})$ by Theorem \ref{TSD-to-quan-rack}. Define a linear map $\varphi:\K[X\times X]\to\K[X]\otimes\K[X]$ as follows:
\begin{eqnarray*}
\varphi(a(x_1,x_2))=ax_1\otimes x_2,\,\,\,\forall a\in\K,(x_1,x_2)\in X\times X.
\end{eqnarray*}
\begin{pro}
With the above notation, $\varphi$ is a linear rack morphism, which makes the following diagram commute:
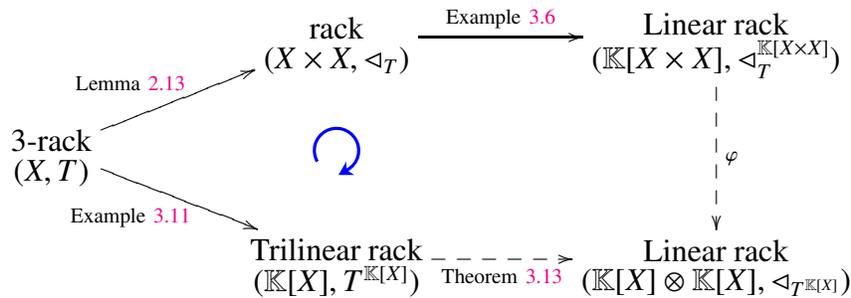
\vspace{-0.5cm}
\begin{center}
\begin{displaymath}
\xymatrix@C=10ex@R=2ex{
  &\txt{\rm rack\\ $(X\times X,\lhd_T)$}
  \ar[r]^-{{\rm Example}~\ref{rack-to-lin-rack}}
  &\txt{\rm Linear~rack\\$(\K[X\times X],\lhd_T^{\K[X\times X]})$}
  \ar@[blue]@{-->}[dd]^-{\varphi}\\
  \txt{\rm $3$-rack\\$(X,T)$}
  \ar[ur]^-{{\rm Lemma}~\ref{3-rack-to-rack}\quad}
  \ar[dr]_-{{\rm Example}~\ref{3-rack-to-tri-rack}\quad}
  &\rotatebox{165}{\color{blue}{\txt{\Huge $\circlearrowright$}}}&\\
  &\txt{\rm Trilinear~rack\\$(\K[X],T^{\K[X]})$}
  \ar@[blue]@{-->}[r]_-{{\rm Theorem}~\ref{TSD-to-quan-rack}}
  &\txt{\rm Linear~rack\\$(\K[X]\otimes\K[X],\lhd_{T^{\K[X]}})$}
}
\end{displaymath}
\end{center}
\end{pro}
\begin{proof}
First, we recall the coproducts and counits on $\K[X\times X]$ and $\K[X]\otimes\K[X]$.
By Example \ref{rack-to-lin-rack}, for any $a(x_1,x_2)\in\K[X\times X]$, we have
\begin{eqnarray*}
\triangle_{\K[X\times X]}\big(a(x,y)\big)
&=&a\triangle_{\K[X\times X]} (x,y) =a(x,y)\otimes (x,y),\\
\varepsilon_{\K[X\times X]}\big(a(x,y)\big)
&=&a\varepsilon_{\K[X\times X]} (x,y) =a.
\end{eqnarray*}
By Proposition \ref{ex-of-coalg}, for any $ax\otimes by\in\K[X]\otimes\K[X]$, we have
\begin{eqnarray*}
\triangle_{\K[X]\otimes\K[X]}(ax\otimes by)
&=&ab\triangle_{\K[X]\otimes\K[X]}(x\otimes y)=ab(x\otimes y)\otimes(x\otimes y),\\
\varepsilon_{\K[X]\otimes\K[X]}(ax\otimes by)
&=&ab\varepsilon_{\K[X]\otimes\K[X]}(x\otimes y)=ab.
\end{eqnarray*}
Then for any $a(x_1,x_2)\in\K[X\times X]$, we have
\begin{eqnarray*}
\triangle_{\K[X]\otimes\K[X]}\circ\varphi\big(a(x,y)\big)
&=&\triangle_{\K[X]\otimes\K[X]}(ax\otimes y)
=a(x\otimes y)\otimes(x\otimes y),\\
(\varphi\otimes\varphi)\circ\triangle_{\K[X\times X]}\big(a(x,y)\big)
&=&(\varphi\otimes\varphi)\big(a(x,y)\otimes (x,y)\big)
=a(x\otimes y)\otimes(x\otimes y),\\
\varepsilon_{\K[X]\otimes\K[X]}\circ\varphi\big(a(x,y)\big)
&=&\varepsilon_{\K[X]\otimes\K[X]}(ax\otimes y)=a,\\
\varepsilon_{\K[X\times X]}\big(a(x,y)\big)&=&a,
\end{eqnarray*}
which implies that $\triangle_{\K[X]\otimes\K[X]}\circ\varphi=(\varphi\otimes\varphi)\circ\triangle_{\K[X\times X]}$ and $\varepsilon_{\K[X]\otimes\K[X]}\circ\varphi=\varepsilon_{\K[X\times X]}$.
Thus $\varphi$ is a coalgebra morphism.

For any $a(x_1,x_2),b(y_1,y_2)\in\K[X\times X]$, we have
\begin{eqnarray*}
\varphi\Big(a(x_1,x_2)\lhd_T^{\K[X\times X]}b(y_1,y_2)\Big)
&=&ab\varphi\Big((x_1,x_2)\lhd_T^{\K[X\times X]}(y_1,y_2)\Big)\\
&=&ab\varphi\Big(T(x_1,y_1,y_2),T(x_2,y_1,y_2)\Big)\\
&=&ab\Big(T(x_1,y_1,y_2)\otimes T(x_2,y_1,y_2)\Big),\\
\varphi\Big(a(x_1,x_2)\Big)\lhd_{T^{\K[X]}}\varphi\Big(b(y_1,y_2)\Big)
&=&\Big(ax_1\otimes x_2\Big)\lhd_{T^{\K[X]}}\Big(by_1\otimes y_2\Big)\\
&=&ab\Big(x_1\otimes x_2\Big)\lhd_{T^{\K[X]}}\Big(y_1\otimes y_2\Big)\\
&=&ab\Big(T(x_1,y_1,y_2)\otimes T(x_2,y_1,y_2)\Big),
\end{eqnarray*}
which implies that $\varphi\Big(a(x_1,x_2)\lhd_T^{\K[X\times X]}b(y_1,y_2)\Big)=\varphi\Big(a(x_1,x_2)\Big)\lhd_{T^{\K[X]}}\varphi\Big(b(y_1,y_2)\Big)$.
Therefore, $\varphi$ is a linear rack morphism from $(\K[X\times X],\lhd_T^{\K[X\times X]})$ to $(\K[X]\otimes\K[X],\lhd_{T^{\K[X]}})$.
\end{proof}

\section{A passage from 3-Leibniz algebras to trilinear racks}\label{sec-3-Lei-tri}

By Lemma \ref{cen-Lei-to-lin-rack}, we have seen that starting from a Leibniz algebra $\huaE$, there is a linear rack structure on the trivial central extension $\K\oplus\huaE$ of $\huaE$. In this section, we show that there is a similar relationship between $3$-Leibniz algebras and trilinear racks. Consequently, a $3$-Leibniz algebra $\huaL$ gives rise to a solution of the Yang-Baxter equation on $(\K\oplus\huaL)\otimes(\K\oplus\huaL)$ through the corresponding trilinear rack by Theorem \ref{TSD-to-solution}.

Let $(\huaL,[\cdot,\cdot,\cdot]_\huaL)$ be a $3$-Leibniz algebra over $\K$. Consider the vector space $\overline{\huaL}\triangleq\K\oplus\huaL$, the following two linear maps $\triangle:\overline{\huaL}\rightarrow\overline{\huaL}\otimes\overline{\huaL}$ and $\varepsilon:\overline{\huaL}\rightarrow\K$ make $(\overline{\huaL},\triangle,\varepsilon)$ a cocommutative coassociative counital coalgebra:
\begin{equation}\label{coproduct-on-overline-L}
\begin{aligned}
\triangle(a,x)&=a(1_\K,0)\otimes(1_\K,0)+(0,x)\otimes(1_\K,0)+(1_\K,0)\otimes(0,x)\\
&=(a,x)\otimes(1_\K,0)+(1_\K,0)\otimes(0,x),\\
\varepsilon(a,x)&=a,
\end{aligned}
\end{equation}
where $(a,x)\in\overline{\huaL}$.
Define a linear map $T:\overline{\huaL}\otimes\overline{\huaL}\otimes\overline{\huaL}\rightarrow\overline{\huaL}$ as follows:
\begin{eqnarray}\label{3-lei-to-TSD-formula}
T\Big((a,x),(b,y),(c,z)\Big)=(abc,bcx+[x,y,z]_\huaL),\,\,\,
\forall (a,x),(b,y),(c,z)\in\overline{\huaL}.
\end{eqnarray}

\begin{thm}\label{3-Lei-to-TSD}
Let $(\huaL,[\cdot,\cdot,\cdot]_\huaL)$ be a $3$-Leibniz algebra. Then $(\overline{\huaL},T)$ is a trilinear rack.
\end{thm}

\begin{proof}
\emptycomment{For any $(a,x)\in\K\oplus\huaL$, we have
\begin{eqnarray*}
(\triangle\otimes{\Id})\circ\triangle(a,x)
&=&(\triangle\otimes{\Id})\Big((a,x)\otimes(1_\K,0)+(1_\K,0)\otimes(0,x)\Big)\\
&=&(a,x)\otimes(1_\K,0)\otimes(1_\K,0)+(1_\K,0)\otimes(0,x)\otimes(1_\K,0)+(1_\K,0)\otimes(0,x)\otimes(1_\K,0)\\
&=&({\Id}\otimes\triangle)\circ\triangle(a,x),\\
(\varepsilon\otimes\varepsilon)\triangle(a,x)
&=&(\varepsilon\otimes\varepsilon)\Big((a,x)\otimes(1_\K,0)+(1_\K,0)\otimes(0,x)\Big)=a
=\varepsilon(a,x),\\
\tau\triangle(a,x)
&=&\tau\Big((a,x)\otimes(1_\K,0)+(1_\K,0)\otimes(0,x)\Big)
=(1_\K,0)\otimes(a,x)+(0,x)\otimes(1_\K,0)\\
&=&a(1_\K,0)\otimes(1_\K,0)+(1_\K,0)\otimes(0,x)+(0,x)\otimes(1_\K,0)\\
&=&(a,x)\otimes(1_\K,0)+(1_\K,0)\otimes(0,x)=\triangle(a,x),
\end{eqnarray*}
which implies that $(\K\oplus\huaL,\triangle,\varepsilon)$ is a cocommutative coassociative counital coalgebra.}
For any $(a,x),(b,y),(c,z),(d,u),(e,v)\in\overline{\huaL}$, we have
\begin{eqnarray*}
&&T\Big(T\Big((a,x),(b,y),(c,z)\Big),(d,u),(e,v)\Big)\\
&=&T\Big((abc,bcx+[x,y,z]_\huaL),(d,u),(e,v)\Big)\\
&=&(abcde,bcdex+de[x,y,z]_\huaL+bc[x,u,v]_\huaL+[[x,y,z]_\huaL,u,v]_\huaL),\\
&&T\Big(T\Big((a,x),(d,u)_{(1)},(e,v)_{(1)}\Big),T\Big((b,y),(d,u)_{(2)},(e,v)_{(2)}\Big),T\Big((c,z),(d,u)_{(3)},(e,v)_{(3)}\Big)\Big)\\
&=&T\Big(T\Big((a,x),(d,u),(e,v)\Big),T\Big((b,y),(1_\K,0),(1_\K,0)\Big),T\Big((c,z),(1_\K,0),(1_\K,0)\Big)\Big)\\
&&+T\Big(T\Big((a,x),(d,u),(1_\K,0)\Big),T\Big((b,y),(1_\K,0),(0,v)\Big),T\Big((c,z),(1_\K,0),(1_\K,0)\Big)\Big)\\
&&+T\Big(T\Big((a,x),(d,u),(1_\K,0)\Big),T\Big((b,y),(1_\K,0),(1_\K,0)\Big),T\Big((c,z),(1_\K,0),(0,v)\Big)\Big)\\
&&+T\Big(T\Big((a,x),(1_\K,0),(e,v)\Big),T\Big((b,y),(0,u),(1_\K,0)\Big),T\Big((c,z),(1_\K,0),(1_\K,0)\Big)\Big)\\
&&+T\Big(T\Big((a,x),(1_\K,0),(1_\K,0)\Big),T\Big((b,y),(0,u),(0,v)\Big),T\Big((c,z),(1_\K,0),(1_\K,0)\Big)\Big)\\
&&+T\Big(T\Big((a,x),(1_\K,0),(1_\K,0)\Big),T\Big((b,y),(0,u),(1_\K,0)\Big),T\Big((c,z),(1_\K,0),(0,v)\Big)\Big)\\
&&+T\Big(T\Big((a,x),(1_\K,0),(e,v)\Big),T\Big((b,y),(1_\K,0),(1_\K,0)\Big),T\Big((c,z),(0,u),(1_\K,0)\Big)\Big)\\
&&+T\Big(T\Big((a,x),(1_\K,0),(1_\K,0)\Big),T\Big((b,y),(1_\K,0),(0,v)\Big),T\Big((c,z),(0,u),(1_\K,0)\Big)\Big)\\
&&+T\Big(T\Big((a,x),(1_\K,0),(1_\K,0)\Big),T\Big((b,y),(1_\K,0),(1_\K,0)\Big),T\Big((c,z),(0,u),(0,v)\Big)\Big)\\
&=&T\Big((ade,dex+[x,u,v]_\huaL),(b,y),(c,z)\Big)\\
&&+T\Big((a,x),(0,[y,u,v]_\huaL),(c,z)\Big)
+T\Big((a,x),(b,y),(0,[z,u,v]_\huaL\Big)\\
&=&(abcde,bcdex+bc[x,u,v]_\huaL+de[x,y,z]_\huaL+[[x,u,v]_\huaL,y,z]_\huaL)\\
&&+(0,[x,[y,u,v]_\huaL,z]_\huaL)+(0,[x,y,[z,u,v]_\huaL]_\huaL),
\end{eqnarray*}
which implies that $T$ satisfies \eqref{TSD-formula} by \eqref{3-Lei-equation}.
By Proposition \ref{ex-of-coalg},   $(\overline{\huaL}\otimes\overline{\huaL}\otimes\overline{\huaL},\triangle_{3},\varepsilon_{3})$ is a coassociative counital coalgebra, where for any $(a,x)\otimes(b,y)\otimes(c,z)\in\overline{\huaL}\otimes\overline{\huaL}\otimes\overline{\huaL}$, \begin{eqnarray*}
\triangle_3\Big((a,x)\otimes(b,y)\otimes(c,z)\Big)
&=&\Big((a,x)_{(1)}\otimes(b,y)_{(1)}\otimes(c,z)_{(1)}\Big)\otimes\Big((a,x)_{(2)}\otimes(b,y)_{(2)}\otimes(c,z)_{(2)}\Big)\\
\varepsilon_{3}\Big((a,x)\otimes(b,y)\otimes(c,z)\Big)
&=&\varepsilon(a,x)\varepsilon(b,y)\varepsilon(c,z).
\end{eqnarray*}
Then by \eqref{coproduct-on-overline-L} and \eqref{3-lei-to-TSD-formula}, we have
\begin{eqnarray*}
&&(T\otimes T)\circ\triangle_3\Big((a,x)\otimes(b,y)\otimes(c,z)\Big)\\
&=&(T\otimes T)\Big(\big((a,x)_{(1)}\otimes(b,y)_{(1)}\otimes(c,z)_{(1)}\big)\otimes\big((a,x)_{(2)}\otimes(b,y)_{(2)}\otimes(c,z)_{(2)}\big)\Big)\\
&=&(T\otimes T)\Big(\big((a,x)\otimes(b,y)\otimes(c,z)\big)\otimes\big((1_\K,0)\otimes(1_\K,0)\otimes(1_\K,0)\big)\Big)\\
&&+(T\otimes T)\Big(\big((a,x)\otimes(b,y)\otimes(1_\K,0)\big)\otimes\big((1_\K,0)\otimes(1_\K,0)\otimes(0,z)\big)\Big)\\
&&+(T\otimes T)\Big(\big((a,x)\otimes(1_\K,0)\otimes(c,z)\big)\otimes\big((1_\K,0)\otimes(0,y)\otimes(1_\K,0)\big)\Big)\\
&&+(T\otimes T)\Big(\big((a,x)\otimes(1_\K,0)\otimes(1_\K,0)\big)\otimes\big((1_\K,0)\otimes(0,y)\otimes(0,z)\big)\Big)\\
&&+(T\otimes T)\Big(\big((1_\K,0)\otimes(b,y)\otimes(c,z)\big)\otimes\big((0,x)\otimes(1_\K,0)\otimes(1_\K,0)\big)\Big)\\
&&+(T\otimes T)\Big(\big((1_\K,0)\otimes(b,y)\otimes(1_\K,0)\big)\otimes\big((0,x)\otimes(1_\K,0)\otimes(0,z)\big)\Big)\\
&&+(T\otimes T)\Big(\big((1_\K,0)\otimes(1_\K,0)\otimes(c,z)\big)\otimes\big((0,x)\otimes(0,y)\otimes(1_\K,0)\big)\Big)\\
&&+(T\otimes T)\Big(\big((1_\K,0)\otimes(1_\K,0)\otimes(1_\K,0)\big)\otimes\big((0,x)\otimes(0,y)\otimes(0,z)\big)\Big)\\
&=&(abc,bcx+[x,y,z]_\huaL)\otimes(1_\K,0)+(1_\K,0)\otimes(0,bcx)+(1_\K,0)\otimes(0,[x,y,z]_\huaL)\\
&=&(abc,bcx+[x,y,z]_\huaL)\otimes(1_\K,0)+(1_\K,0)\otimes(0,bcx+[x,y,z]_\huaL),\\
&=&\triangle(abc,bcx+[x,y,z]_\huaL)\\
&=&\triangle\circ T\Big((a,x),(b,y),(c,z)\Big),\\
&&\varepsilon\circ T\Big((a,x),(b,y),(c,z)\Big)\\
&=&\varepsilon(abc,bcx+[x,y,z]_\huaL)=abc
=\varepsilon(a,x)\varepsilon(b,y)\varepsilon(c,z)\\
&=&\varepsilon_3\Big((a,x)\otimes(b,y)\otimes(c,z)\Big),
\end{eqnarray*}
which implies that $T$ is a coalgebra morphism. Therefore, $(\overline{\huaL},T)$ is a trilinear shelf.

Moreover, define $\widetilde{T}:\overline{\huaL}\otimes\overline{\huaL}\otimes\overline{\huaL}\rightarrow\overline{\huaL}$
as follows:
\begin{eqnarray}\label{rever-of-T}
\widetilde{T}\Big((a,x),(b,y),(c,z)\Big)=(abc,bcx-[x,z,y]_\huaL).
\end{eqnarray}
Similarly, we obtain that $(\overline{\huaL},\widetilde{T})$ is also a trilinear shelf. For any $(a,x),(b,y),(c,z)\in\overline{\huaL}$, we have
\begin{eqnarray*}
&&\widetilde{T}\Big(T\Big((a,x),(b,y)_{(2)},(c,z)_{(2)}\Big),(c,z)_{(1)},(b,y)_{(1)}\Big)\\
&=&\widetilde{T}\Big(T\Big((a,x),(1_\K,0),(1_\K,0)\Big),(c,z),(b,y)\Big)
+\widetilde{T}\Big(T\Big((a,x),(1_\K,0),(0,z)\Big),(1_\K,0),(b,y)\Big)\\
&&+\widetilde{T}\Big(T\Big((a,x),(0,y),(1_\K,0)\Big),(c,z),(1_\K,0)\Big)
+\widetilde{T}\Big(T\Big((a,x),(0,y),(0,z)\Big),(1_\K,0),(1_\K,0)\Big)\\
&=&\widetilde{T}\Big((a,x),(c,z),(b,y)\Big)+\widetilde{T}\Big((0,[x,y,z]_\huaL),(1_\K,0),(1_\K,0)\Big)\\
&=&(abc,bcx-[x,y,z]_\huaL)+(0,[x,y,z]_\huaL)\\
&=&(abc,bcx)=bc(a,x)=\varepsilon(b,y)\varepsilon(c,z)\cdot(a,x).
\end{eqnarray*}
Similarly, ${T}\Big(\widetilde {T}\Big((a,x),(b,y)_{(2)},(c,z)_{(2)}\Big),(c,z)_{(1)},(b,y)_{(1)}\Big)=\varepsilon(b,y)\varepsilon(c,z)\cdot(a,x)$,
which implies that $(\overline{\huaL},T)$ is a trilinear rack.
\end{proof}

\begin{rmk}
Note that in \cite{Abramov}, the authors constructed a trilinear rack by a $3$-Lie algebra. Here we generalize the construction to a $3$-Leibniz algebra. Furthermore, the $\widetilde {T}$ given by \eqref{rever-of-T} is different from the one given in \cite[Lemma~4.1]{Abramov}.
\end{rmk}

\begin{ex}
Consider the $3$-Leibniz algebra $(\BO,[\cdot,\cdot,\cdot])$ showed in {\rm Example \ref{ex-octonion}}. By {\rm Theorem \ref{3-Lei-to-TSD}}, we obtain a trilinear rack
$(\R\oplus\BO,T)$, where for all $(a,x),(b,y),(c,z)\in\R\oplus\BO$,
$$T\Big((a,x),(b,y),(c,z)\Big)=\Big(abc,bcx+z(yx)-y(zx)+(xy)z-(xz)y+(yx)z-y(xz)\Big).$$
\end{ex}

\begin{ex}
Consider the $3$-Leibniz algebra $(\gl(V)\oplus V,[\cdot,\cdot,\cdot])$ showed in {\rm Example \ref{ex-omni-3-lie}}, where $V$ is a vector space over $\K$.
Then we obtain a trilinear rack
$(\K\oplus\gl(V)\oplus V,T)$ by {\rm Theorem \ref{3-Lei-to-TSD}}, where for all $(a,A,u),(b,B,v),(c,C,w)\in\K\oplus\gl(V)\oplus V$,
\begin{eqnarray*}
T\Big((a,A,u),(b,B,v),(c,C,w)\Big)=(abc,bcA+[A,[B,C]],bcu-[B,C]u).
\end{eqnarray*}
\end{ex}

\begin{pro}\label{3-lei-to-sol}
Let $(\huaL,[\cdot,\cdot,\cdot]_\huaL)$ be a $3$-Leibniz algebra. Then there is a solution of the Yang-Baxter equation $R^{T}_{\overline{\huaL}\otimes\overline{\huaL}}$ on $\overline{\huaL}\otimes\overline{\huaL}$ as follows:
\begin{eqnarray*}
&&R^{T}_{\overline{\huaL}\otimes\overline{\huaL}}
\Big((a_1,x_1)\otimes(a_2,x_2)\otimes(b_1,y_1)\otimes(b_2,y_2)\Big)\\
&=&(b_1,y_1)\otimes(b_2,y_2)\otimes(a_1,x_1)\otimes(a_2,x_2)\\
&&+(1_\mathds{K},0)\otimes(1_\mathds{K},0)\otimes\Big((0,[x_1,y_1,y_2]_\huaL)\otimes(a_2,x_2)+(a_1,x_1)\otimes(0,[x_2,y_1,y_2]_\huaL)\Big).
\end{eqnarray*}
\end{pro}
\begin{proof}
By {\rm Theorem \ref{3-Lei-to-TSD}}, there is a trilinear rack $(\overline{\huaL},T)$ constructed by the $3$-Leibniz algebra $(\huaL,[\cdot,\cdot,\cdot]_\huaL)$, which induces a solution of the Yang-Baxter equation on $\overline{\huaL}\otimes\overline{\huaL}$ by {\rm Theorem \ref{TSD-to-solution}} as follows:
\begin{eqnarray*}
&&R^{T}_{\overline{\huaL}\otimes\overline{\huaL}}
\Big((a_1,x_1)\otimes(a_2,x_2)\otimes(b_1,y_1)\otimes(b_2,y_2)\Big)\\
&=&\Big((b_1,y_1)_{(1)}\otimes (b_2,y_2)_{(1)}\Big)\otimes\Big(T((a_1,x_1),(b_1,y_1)_{(2)},(b_2,y_2)_{(2)})\otimes T((a_2,x_2),(b_1,y_1)_{(3)},(b_2,y_2)_{(3)})\Big)\\
&=&\Big((b_1,y_1)\otimes(b_2,y_2)\Big)\otimes\Big(T\big((a_1,x_1),(1_\K,0),(1_\K,0)\big)\otimes T\big((a_2,x_2),(1_\K,0),(1_\K,0)\big)\Big)\\
&&+\Big((1_\K,0)\otimes(b_2,y_2)\Big)\otimes\Big(T\big((a_1,x_1),(0,y_1),(1_\K,0)\big)\otimes T\big((a_2,x_2),(1_\K,0),(1_\K,0)\big)\Big)\\
&&+\Big((1_\K,0)\otimes(b_2,y_2)\Big)\otimes\Big(T\big((a_1,x_1),(1_\K,0),(1_\K,0)\big)\otimes T\big((a_2,x_2),(0,y_1),(1_\K,0)\big)\Big)\\
&&+\Big((b_1,y_1)\otimes(1_\K,0)\Big)\otimes\Big(T\big((a_1,x_1),(1_\K,0),(0,y_2)\big)\otimes T\big((a_2,x_2),(1_\K,0),(1_\K,0)\big)\Big)\\
&&+\Big((1_\K,0)\otimes(1_\K,0)\Big)\otimes\Big(T\big((a_1,x_1),(0,y_1),(0,y_2)\big)\otimes T\big((a_2,x_2),(1_\K,0),(1_\K,0)\big)\Big)\\
&&+\Big((1_\K,0)\otimes(1_\K,0)\Big)\otimes\Big(T\big((a_1,x_1),(1_\K,0),(0,y_2)\big)\otimes T\big((a_2,x_2),(0,y_1),(1_\K,0)\big)\Big)\\
&&+\Big((b_1,y_1)\otimes(1_\K,0)\Big)\otimes\Big(T\big((a_1,x_1),(1_\K,0),(1_\K,0)\big)\otimes T\big((a_2,x_2),(1_\K,0),(0,y_2)\big)\Big)\\
&&+\Big((1_\K,0)\otimes(1_\K,0)\Big)\otimes\Big(T\big((a_1,x_1),(0,y_1),(1_\K,0)\big)\otimes T\big((a_2,x_2),(1_\K,0),(0,y_2)\big)\Big)\\
&&+\Big((1_\K,0)\otimes(1_\K,0)\Big)\otimes\Big(T\big((a_1,x_1),(1_\K,0),(1_\K,0)\big)\otimes T\big((a_2,x_2),(0,y_1),(0,y_2)\big)\Big)\\
&=&\Big((b_1,y_1)\otimes(b_2,y_2)\Big)\otimes\Big((a_1,x_1)\otimes(a_2,x_2)\Big)\\
&&+\Big((1_\K,0)\otimes(1_\K,0)\Big)\otimes\Big((0,[x_1,y_1,y_2]_\huaL)\otimes(a_2,x_2)\Big)\\
&&+\Big((1_\K,0)\otimes(1_\K,0)\Big)\otimes\Big((a_1,x_1)\otimes(0,[x_2,y_1,y_2]_\huaL)\Big)\\
&=&(b_1,y_1)\otimes(b_2,y_2)\otimes(a_1,x_1)\otimes(a_2,x_2)\\
&&+(1_\mathds{K},0)\otimes(1_\mathds{K},0)\otimes\Big((0,[x_1,y_1,y_2]_\huaL)\otimes(a_2,x_2)+(a_1,x_1)\otimes(0,[x_2,y_1,y_2]_\huaL)\Big).\qedhere
\end{eqnarray*}
\end{proof}

\section{Central extensions of (3-)Leibniz algebras and the Yang-Baxter equation}\label{sec-lei-3-lei}

Recall that a central Leibniz algebra gives rise to a solution of the Yang-Baxter equation (Theorem \ref{lei to solution}). In this section, first we use central extensions of Leibniz algebras to construct solutions of the Yang-Baxter equation. Then  we prove that  a $3$-Leibniz algebra $\huaL$ gives rise to two central Leibniz algebras, which are  $(\K\oplus\huaL)\otimes(\K\oplus\huaL)$ and  $\K\oplus(\huaL\otimes\huaL)$ respectively. We show that there is a homomorphism between the corresponding solutions of the Yang-Baxter equation.
\subsection{Leibniz algebras and the Yang-Baxter equation}
A central Leibniz algebra $(\huaE,[\cdot,\cdot]_\huaE,{\bf 1})$ can be viewed as a central extension of Leibniz algebras. Recall that a central extension of a Leibniz algebra $(\huaE,[\cdot,\cdot]_\huaE)$ is a short exact sequence of Leibniz algebras
$$0\longrightarrow \K\stackrel{i}{\longrightarrow} \hat{\huaE}\stackrel{p}{\longrightarrow} \huaE\longrightarrow 0,$$
such that ${i}(\K)$ is in the center of $\hat{\huaE}$. Then we can see that a central Leibniz algebra $(\huaE,[\cdot,\cdot]_\huaE,{\bf 1})$ gives rise to a central extension of $\huaE/{\K{\bf 1}}$:
$$0\longrightarrow\K{\bf 1}\longrightarrow \huaE\longrightarrow{\huaE/{\K{\bf 1}}}\longrightarrow 0.$$

 Two such extensions
$\hat{\huaE}_1$ and $\hat{\huaE}_2$ of a Leibniz algebra $(\huaE,[\cdot,\cdot]_\huaE)$ by $\K$ are isomorphic if there exists a Leibniz algebra homomorphism $\theta$ from $\hat{\huaE}_1$ to $\hat{\huaE}_2$ such that
$$\theta\circ i_1=i_2,\,\,\,p_2\circ\theta=p_1.$$
It is well known that central extensions of  a Leibniz algebra $(\huaE,[\cdot,\cdot]_\huaE)$ by $\K$  are classified by the second cohomology group $\huaH^2(\huaE;\K)$ of the Leibniz algebra $(\huaE,[\cdot,\cdot]_\huaE)$ with coefficients in the trivial representations \cite{Loday}. In fact, suppose that $s:\huaE\rightarrow\hat{\huaE}$ is a linear map satisfying $p\circ s={\Id}_\huaE$, then $\hat{\huaE}$ is isomorphic to $\K\oplus\huaE$ as vector space by $F(X)=(X-sp(X),p(X))$ for all $X\in\hat{\huaE}$ and $F^{-1}(a,x)=a+s(x)$ for all $(a,x)\in\K\oplus\huaE$. Define a linear map $\omega:\huaE\otimes\huaE\rightarrow\K$ as follows:
\begin{eqnarray*}
\omega(x,y)=[s(x),s(y)]_{\hat{\huaE}}-s[x,y]_\huaE,\,\,\,\forall x,y\in\huaE.
\end{eqnarray*}
Then we obtain a bilinear map $[\cdot,\cdot]_\omega$ on $\K\oplus\huaE$ by
\begin{eqnarray}\label{cen-ext-bracket}
[(a,x),(b,y)]_\omega=F[F^{-1}(a,x),F^{-1}(b,y)]_{\hat{\huaE}}=(\omega(x,y),[x,y]_\huaE).
\end{eqnarray}
$(\K\oplus\huaE,[\cdot,\cdot]_\omega)$ is a Leibniz algebra if and only if $\omega$ is a 2-cocycle with coefficients in the trivial representation $\K$. Moreover, $[\omega]\in\huaH^2(\huaE;\K)$ is independent of the choice of $s$.
\emptycomment{
\begin{proof}
$(\huaV\oplus\huaE,[\cdot,\cdot]_\omega)$ is a Leibniz algebra if and only if
\begin{eqnarray*}
[[(a,x),(b,y)]_\omega,(c,z)]_\omega=[[(a,x),(c,z)]_\omega,(b,y)]_\omega+[(a,x),[(b,y),(c,z)]_\omega]_\omega,
\end{eqnarray*}
which is equivalent to
\begin{eqnarray*}
\omega([x,y]_\huaE,z)-\omega([x,z]_\huaE,y)-\omega(x,[y,z]_\huaE)=0.
\end{eqnarray*}
It shows that $\partial\omega(x,y,z)=0$, where $\partial$ is the coboundary operator with respect to the trivial representation $\huaV$. Thus, $(\huaV\oplus\huaE,[\cdot,\cdot]_\omega)$ is a Leibniz algebra if and only if $\omega$ is a 2-cocycle.

Let $s':\huaE\rightarrow{\hat{\huaE}}$ be another section of the central extension of $\huaE$ by $\huaV$ and $\omega'$ the associated $2$-cocycle. By $p\circ s'=\Id_{\huaE}=$ $p\circ s$, we obtain that $(s'-s)(x)\in$ Ker$p\cong$ Im$i=\huaV$. Then we define $f=s'-s\in$ Hom$(\huaE,\huaV)=\mathcal{C}^1(\huaE;\huaV)$. By direct computation, we have
\begin{eqnarray*}
(\omega'-\omega)(x,y)&=&[s'(x),s'(y)]_{\tilde{\huaE}}-s'[x,y]_\huaE-[s(x),s(y)]_{\tilde{\huaE}}+s[x,y]_\huaE\\
&=&[s'(x),s'(y)-s(y)]_{\tilde{\huaE}}+[s'(x)-s(x),s(y)]_{\tilde{\huaE}}-s'[x,y]_\huaE+s[x,y]_\huaE\\
&=&-f[x,y]_\huaE=\partial f(x,y),
\end{eqnarray*}
which implies that $\omega$ and $\omega'$ are in the same cohomology class, i.e. $[\omega]\in\huaH^2(\huaE;\huaV)$ is independent of the choice of sections.
\end{proof}}
Since any central extension of $\huaE$ by $\K$ is isomorphic to $(\K\oplus\huaE,[\cdot,\cdot]_\omega)$, then we only consider central extensions of the form $(\K\oplus\huaE,[\cdot,\cdot]_\omega)$ in the sequel.

\begin{thm}\label{cen-ext-to-sol}
Let $(\huaE,[\cdot,\cdot]_\huaE)$ be a Leibniz algebra and $(\K\oplus\huaE,[\cdot,\cdot]_\omega)$ a central extension of $\huaE$ by $\K$. Then the linear map $R^{Lei}:(\K\oplus\huaE)\otimes(\K\oplus\huaE)\rightarrow(\K\oplus\huaE)\otimes(\K\oplus\huaE)$ defined as follows is a solution of the Yang-Baxter equation:
\begin{eqnarray}\label{Lei-cen-ext-to-solution}
&&R^{Lei}((a,x)\otimes(b,y))=(b,y)\otimes(a,x)+(1_\K,0)\otimes(\omega(x,y),[x,y]_\huaE).
\end{eqnarray}
Moreover,  solutions induced by  isomorphic central extensions of $\huaE$ are equivalent.
\end{thm}

\begin{proof}
It is obvious that $(1_\K,0)$ is a central element of the Leibniz algebra $\K\oplus\huaE$, then we obtain a solution of the Yang-Baxter equation $R^{Lei}:(\K\oplus\huaE)\otimes(\K\oplus\huaE)\rightarrow(\K\oplus\huaE)\otimes(\K\oplus\huaE)$
given by \eqref{lei to solution formula} as follows:
\begin{eqnarray*}\label{Lei-cen-ext-to-solution}
R^{Lei}((a,x)\otimes(b,y))&=&(b,y)\otimes(a,x)+(1_\K,0)\otimes[(a,x),(b,y)]_\omega\\
&=&(b,y)\otimes(a,x)+(1_\K,0)\otimes(\omega(x,y),[x,y]_\huaE).
\end{eqnarray*}

Let $\theta:(\K\oplus\huaE)_{\omega_1}\rightarrow(\K\oplus\huaE)_{\omega_2}$ be the Leibniz algebra homomorphism such that two central extensions $(\K\oplus\huaE,[\cdot,\cdot]_{\omega_1})$ and $(\K\oplus\huaE,[\cdot,\cdot]_{\omega_2})$ are isomorphic. Then  we have $\theta(a,x)=(a-f(x),x)$, where $f:\huaE\rightarrow\K$ such that $\omega_2-\omega_1=\partial f$, where $\partial:\Hom(\otimes^k\huaE,\K)\to \Hom(\otimes^{k+1}\huaE,\K)$ is the coboundary operator given by
\begin{eqnarray*}
(\partial f)(x_1,\cdots,x_{k+1})
&=&\sum_{1\le i<j\le k+1}(-1)^{j+1}f(x_{1},\cdots,x_{i-1},[x_i,x_j]_\huaE,x_{i+1},\cdots,\hat{x_j},\cdots,x_{k+1}).
\end{eqnarray*}
Denote by $R^{Lei}_1$ and $R^{Lei}_2$the corresponding solutions of the Yang-Baxter equation induced by central Leibniz algebras $(\K\oplus\huaE,[\cdot,\cdot]_{\omega_1})$ and $(\K\oplus\huaE,[\cdot,\cdot]_{\omega_2})$ respectively. For all $a,b\in\K,\,x,y\in\huaE$, we have
\begin{eqnarray*}
(\theta\otimes\theta)R^{Lei}_1((a,x)\otimes(b,y))&=&(\theta\otimes\theta)((b,y)\otimes(a,x)+(1_\K,0)\otimes(\omega_1(x,y),[x,y]_\huaE))\\
&=&(b-f(y),y)\otimes(a-f(x),x)+(1_\K,0)\otimes(\omega_1(x,y)-f[x,y]_\huaE,[x,y]_\huaE),\\
R^{Lei}_2(\theta\otimes\theta)((a,x)\otimes(b,y))&=&R^{Lei}_2((a-f(x),x)\otimes(b-f(y),y))\\
&=&(b-f(y),y)\otimes(a-f(x),x)+(1_\K,0)\otimes(\omega_2(x,y),[x,y]_\huaE),
\end{eqnarray*}
which implies that $(\theta\otimes\theta)R^{Lei}_1=R^{Lei}_2(\theta\otimes\theta)$, i.e. $R^{Lei}_1$ and $R^{Lei}_2$ are equivalent.
\end{proof}

\begin{ex}
Consider the $1$-dimensional abelian Leibniz algebra $\huaE$ over $K$ with the basis  $\{e\}$. Define $\omega:\huaE\otimes\huaE\to\K$ by $\omega(e,e)=1_\K$, which is obviously a $2$-cocycle with coefficients in the trivial representation $\K$. By {\rm Theorem \ref{cen-ext-to-sol}}, we obtain a solution of the Yang-Baxter equation $R^{Lei}_{\K\oplus\huaE}$ given with respect to the basis $\{(1_\K,0)\otimes(1_\K,0),(1_\K,0)\otimes(0,e),(0,e)\otimes(1_\K,0),(0,e)\otimes(0,e)\}$ of $(\K\oplus\huaE)\otimes(\K\oplus\huaE)$ as follows:
\begin{equation*}
R^{Lei}_{\K\oplus\huaE}=
\left(
\begin{array}{cccc}
1 & 0 & 0 & 1\\
0 & 0 & 1 & 0\\
0 & 1 & 0 & 0\\
0 & 0 & 0 & 1
\end{array}
\right).
\end{equation*}
\end{ex}

\begin{ex}
In \cite{AOR}, the author showed that there are $4$ type of $2$-dimensional Leibniz algebras over $\C$ up to isomorphism with respect to the basis $\{e_1,e_2\}$:
\begin{eqnarray*}
&&\huaE_1:\quad abelian~algebra,\\
&&\huaE_2:\quad [e_1,e_2]_{\huaE_2}=e_2=-[e_2,e_1]_{\huaE_2},\\
&&\huaE_3:\quad [e_2,e_2]_{\huaE_3}=e_1,\\
&&\huaE_4:\quad [e_1,e_2]_{\huaE_4}=[e_2,e_2]_{\huaE_4}=e_1.
\end{eqnarray*}
Since $\{e_1,e_2\}$ is the basis of $\huaE_i$, $\{(1,0)\otimes(1,0),(1,0)\otimes(0,e_1),(1,0)\otimes(0,e_2),\cdots,(0,e_2)\otimes(0,e_2)\}$ is a basis of $(\C\oplus\huaE_i)\otimes(\C\oplus\huaE_i)$. By {\rm Theorem \ref{cen-ext-to-sol}}, we obtain solutions of the Yang-Baxter equation $R^{Lei}_{\C\oplus\huaE_i}$ given with respect to the above basis as follows:
\begin{equation*}
R_{\C\oplus\huaE_1}^{Lei}=
\left(
\begin{array}{ccccccccc}
    1 & 0 & 0 & 0 & a_1 & a_2 & 0 & a_3 & a_4 \\
    0 & 0 & 0 & 1 & 0 & 0 & 0 & 0 & 0 \\
    0 & 0 & 0 & 0 & 0 & 0 & 1 & 0 & 0 \\
    0 & 1 & 0 & 0 & 0 & 0 & 0 & 0 & 0 \\
    0 & 0 & 0 & 0 & 1 & 0 & 0 & 0 & 0  \\
    0 & 0 & 0 & 0 & 0 & 0 & 0 & 1 & 0 \\
    0 & 0 & 1 & 0 & 0 & 0 & 0 & 0 & 0  \\
    0 & 0 & 0 & 0 & 0 & 1 & 0 & 0 & 0 \\
    0 & 0 & 0 & 0 & 0 & 0 & 0 & 0 & 1
\end{array}
\right),\quad
R_{\C\oplus\huaE_2}^{Lei}=
\left(
\begin{array}{ccccccccc}
    1 & 0 & 0 & 0 & b_1 & b_2 & 0 & 0 & 0 \\
    0 & 0 & 0 & 1 & 0 & 0 & 0 & 0 & 0 \\
    0 & 0 & 0 & 0 & 0 & 1 & 1 & -1 & 0 \\
    0 & 1 & 0 & 0 & 0 & 0 & 0 & 0 & 0 \\
    0 & 0 & 0 & 0 & 1 & 0 & 0 & 0 & 0  \\
    0 & 0 & 0 & 0 & 0 & 0 & 0 & 1 & 0 \\
    0 & 0 & 1 & 0 & 0 & 0 & 0 & 0 & 0  \\
    0 & 0 & 0 & 0 & 0 & 1 & 0 & 0 & 0 \\
    0 & 0 & 0 & 0 & 0 & 0 & 0 & 0 & 1
\end{array}
\right),
\end{equation*}
\begin{equation*}
R_{\C\oplus\huaE_3}^{Lei}=
\left(
\begin{array}{ccccccccc}
    1 & 0 & 0 & 0 & 0 & 0 & 0 & c_1 & c_2 \\
    0 & 0 & 0 & 1 & 0 & 0 & 0 & 0 & 1 \\
    0 & 0 & 0 & 0 & 0 & 0 & 1 & 0 & 0 \\
    0 & 1 & 0 & 0 & 0 & 0 & 0 & 0 & 0 \\
    0 & 0 & 0 & 0 & 1 & 0 & 0 & 0 & 0  \\
    0 & 0 & 0 & 0 & 0 & 0 & 0 & 1 & 0 \\
    0 & 0 & 1 & 0 & 0 & 0 & 0 & 0 & 0  \\
    0 & 0 & 0 & 0 & 0 & 1 & 0 & 0 & 0 \\
    0 & 0 & 0 & 0 & 0 & 0 & 0 & 0 & 1
\end{array}
\right),\quad
R_{\C\oplus\huaE_4}^{Lei}=
\left(
\begin{array}{ccccccccc}
    1 & 0 & 0 & 0 & 0 & 0 & 0 & d_1 & d_2 \\
    0 & 0 & 0 & 1 & 0 & 1 & 0 & 0 & 1 \\
    0 & 0 & 0 & 0 & 0 & 0 & 1 & 0 & 0 \\
    0 & 1 & 0 & 0 & 0 & 0 & 0 & 0 & 0 \\
    0 & 0 & 0 & 0 & 1 & 0 & 0 & 0 & 0  \\
    0 & 0 & 0 & 0 & 0 & 0 & 0 & 1 & 0 \\
    0 & 0 & 1 & 0 & 0 & 0 & 0 & 0 & 0  \\
    0 & 0 & 0 & 0 & 0 & 1 & 0 & 0 & 0 \\
    0 & 0 & 0 & 0 & 0 & 0 & 0 & 0 & 1
\end{array}
\right),
\end{equation*}
where $a_i,b_j,c_j,d_j\in\C$ for $i=1,2,3,4$ and $j=1,2$.
\emptycomment{
Consider the trivial representation $\mathds{C}$, we obtain the $2$-cocycle $\omega_i:\huaE_i\otimes\huaE_i\rightarrow\mathds{C}$ corresponding to the central extension $\mathds{C}\oplus\huaE_i$ for each Leibniz algebra $\huaE_i$ as follows:
\begin{itemize}
  \item $\omega_1(e_i,e_j)=a_{ij}$
  \item $\begin{cases}
        \omega_2(e_1,e_1)=b_1& \\
        \omega_2(e_1,e_2)=b_2& \\
        \omega_2(e_2,e_1)=\omega_2(e_2,e_2)=0
        \end{cases}$
  \item $\begin{cases}
        \omega_3(e_1,e_1)=0& \\
        \omega_3(e_1,e_2)=0& \\
        \omega_2(e_2,e_1)=c_1 &\\
        \omega_3(e_2,e_2)=c_2
        \end{cases}$
  \item $\begin{cases}
        \omega_4(e_1,e_1)=0& \\
        \omega_4(e_1,e_2)=0& \\
        \omega_4(e_2,e_1)=d_1 &\\
        \omega_4(e_2,e_2)=d_2
        \end{cases}$
\end{itemize}
Then by \eqref{Lei-cen-ext-to-solution}, we can obtain the concrete Yang-Baxter operator $R^{Lei}_i:(\mathds{C}\oplus\huaE_i)^{\otimes 2}\rightarrow(\mathds{C}\oplus\huaE_i)^{\otimes 2}$ for $1\leq i\leq 4$.}

\end{ex}
\subsection{$3$-Leibniz algebras and the Yang-Baxter equation}
\begin{defi}
Let $(\huaL,[\cdot,\cdot,\cdot]_\huaL)$ be a $3$-Leibniz algebra.
An element ${\bf 1}\in\huaL$ is called a {\bf central element} of the $3$-Leibniz algebra $\huaL$ if for any $x,y\in\huaL$,
\begin{eqnarray*}
&&[{\bf 1},x,y]_\huaL=[x,{\bf 1},y]_\huaL=[x,y,{\bf 1}]_\huaL=0.
\end{eqnarray*}
Then $(\huaL,[\cdot,\cdot,\cdot]_\huaL,{\bf 1})$ is called a {\bf central $3$-Leibniz algebra}.
\end{defi}

\begin{thm}\label{3-Lei-to-Lei-and-sol}
Let $(\huaL,[\cdot,\cdot,\cdot]_\huaL,{\bf 1})$ be a central $3$-Leibniz algebra. Then $(\huaL\otimes\huaL,\{\cdot,\cdot\},{\bf 1}\otimes{\bf 1})$ is a central Leibniz algebra, where the Leibniz algebra structure $\{\cdot,\cdot\}$ on $\huaL\otimes\huaL$ is defined by \eqref{fundamental-element-formula}. Consequently, a central $3$-Leibniz algebra gives rise to a solution of the Yang-Baxter equation, illustrated by the following diagram:
\vspace{-0.5cm}
\begin{center}
\begin{displaymath}
\xymatrix@C=10ex{
  \txt{\rm central $3$-Leibniz algebra~\\$(\huaL,[\cdot,\cdot,\cdot]_\huaL,{\bf 1})$}
   \ar@[blue]@{-->}[r]
  &\txt{{\rm central Leibniz algebra}~\\$(\huaL\otimes\huaL,\{\cdot,\cdot\},{\bf 1}\otimes{\bf 1})$}
  \ar[r]^-{{\rm Lemma}~\ref{lei to solution}}
  &\txt{{$R^{Lei}_{\huaL\otimes\huaL}$}}
}
\end{displaymath}
\end{center}
\end{thm}

\begin{proof}

We only need to show that ${\bf 1}\otimes{\bf 1}$ is a central element of the Leibniz algebra $(\huaL,\{\cdot,\cdot\})$.
Since ${\bf 1}$ is a central element of the $3$-Leibniz algebra $\huaL$, then for any $x_1\otimes x_2\in\huaL\otimes\huaL$, we have
\begin{eqnarray*}
&&\{{\bf 1}\otimes{\bf 1},x_1\otimes x_2\}
=[{\bf 1},x_1,x_2]_\huaL\otimes {\bf 1}+{\bf 1}\otimes[{\bf 1},x_1,x_2]_\huaL
=0,\\
&&\{x_1\otimes x_2,{\bf 1}\otimes{\bf 1}\}
=[x_1,{\bf 1},{\bf 1}]_\huaL\otimes x_2+x_1\otimes[x_2,{\bf 1},{\bf 1}]_\huaL
=0,
\end{eqnarray*}
which implies that ${\bf 1}\otimes{\bf 1}$ is a central element of the Leibniz algebra $(\huaL\otimes\huaL,\{\cdot,\cdot\})$. Then by Theorem \ref{lei to solution}, there is a solution of the Yang-Baxter equation $R^{Lei}_{\huaL\otimes\huaL}:
\Big(\huaL\otimes\huaL\Big)\otimes\Big(\huaL\otimes\huaL\Big)\rightarrow\Big(\huaL\otimes\huaL\Big)\otimes\Big(\huaL\otimes\huaL\Big)$
defined as follows:
\begin{eqnarray*}
\quad\quad&&R^{Lei}_{\huaL\otimes\huaL}\Big((x_1\otimes x_2)\otimes(y_1\otimes y_2)\Big)\\
&=&(y_1\otimes y_2)\otimes(x_1\otimes x_2)
+({\bf 1\otimes 1})\otimes([x_1,y_1,y_2]_\huaL\otimes x_2)
+({\bf 1\otimes 1})\otimes(x_1\otimes[x_2,y_1,y_2]_\huaL).\,\,\,\qedhere
\end{eqnarray*}
\end{proof}

Similar to the case of Leibniz algebras, for any $3$-Leibniz algebra, we consider a central extension of $3$-Leibniz algebras, which will give rise to a central 3-Leibniz algebras. Therefore we obtain a solution of the Yang-Baxter equation by Theorem \ref{3-Lei-to-Lei-and-sol}.

\begin{cor}\label{3-lei-cen-ext-to-sol}
Let $(\huaL,[\cdot,\cdot,\cdot]_\huaL)$ be a $3$-Leibniz algebra over $\mathds{K}$.
Then we obtain a solution of the Yang-Baxter equation $R^{Lei}_{\overline{\huaL}\otimes\overline{\huaL}}$ on $\overline{\huaL}\otimes\overline{\huaL}$ as follows, where $\overline{\huaL}=\mathds{K}\oplus\huaL$:
\begin{eqnarray*}
&&R^{Lei}_{\overline{\huaL}\otimes\overline{\huaL}}
\Big((a_1,x_1)\otimes(a_2,x_2)\otimes(b_1,y_1)\otimes(b_2,y_2)\Big)\\
&=&(b_1,y_1)\otimes(b_2,y_2)\otimes(a_1,x_1)\otimes(a_2,x_2)\\
&&+(1_\mathds{K},0)\otimes(1_\mathds{K},0)\otimes\Big((0,[x_1,y_1,y_2]_\huaL)\otimes(a_2,x_2)+(a_1,x_1)\otimes(0,[x_2,y_1,y_2]_\huaL)\Big).
\end{eqnarray*}
\end{cor}

\begin{proof}
Define a linear map $[\cdot,\cdot,\cdot]_{\overline{\huaL}}:\overline{\huaL}\otimes\overline{\huaL}\to\overline{\huaL}$ as follows:
\begin{eqnarray*}
[(a,x),(b,y),(c,z)]_{\overline{\huaL}}=(0,[x,y,z]_\huaL),
\,\,\,\forall (a,x),(b,y),(c,z)\in\overline{\huaL}.
\end{eqnarray*}
By direct computation, we obtain that
$(\overline{\huaL},[\cdot,\cdot,\cdot]_{\overline{\huaL}},(1_\mathds{K},0))$
is a central $3$-Leibniz algebra.
By {\rm Theorem \ref{3-Lei-to-Lei-and-sol}}, there is a central Leibniz algebra $(\overline{\huaL}\otimes\overline{\huaL},\{\cdot,\cdot\},(1_\mathds{K},0)\otimes(1_\mathds{K},0))$  and the induced solution of the Yang-Baxter equation is
\begin{eqnarray*}
\,\,\quad\quad&&R^{Lei}_{\overline{\huaL}\otimes\overline{\huaL}}
\Big((a_1,x_1)\otimes(a_2,x_2)\otimes(b_1,y_1)\otimes(b_2,y_2)\Big)\\
&=&(b_1,y_1)\otimes(b_2,y_2)\otimes(a_1,x_1)\otimes(a_2,x_2)\\
&&+{(1_\mathds{K},0)\otimes(1_\mathds{K},0)}\otimes[(a_1,x_1),(b_1,y_1),(b_2,y_2)]_{\overline{\huaL}}\otimes (a_2,x_2)\\
&&+{(1_\mathds{K},0)\otimes(1_\mathds{K},0)}\otimes(a_1,x_1)\otimes[(a_2,x_2),(b_1,y_1),(b_2,y_2)]_{\overline{\huaL}}\\
&=&(b_1,y_1)\otimes(b_2,y_2)\otimes(a_1,x_1)\otimes(a_2,x_2)\\
&&+(1_\mathds{K},0)\otimes(1_\mathds{K},0)\otimes\Big((0,[x_1,y_1,y_2]_\huaL)\otimes(a_2,x_2)+(a_1,x_1)\otimes(0,[x_2,y_1,y_2]_\huaL)\Big).\quad\quad\qedhere
\end{eqnarray*}
\end{proof}

It is obvious that the above solution of the Yang-Baxter equation is consistent with the one given in {\rm Proposition \ref{3-lei-to-sol}}, that is, we have the following commutative diagram:
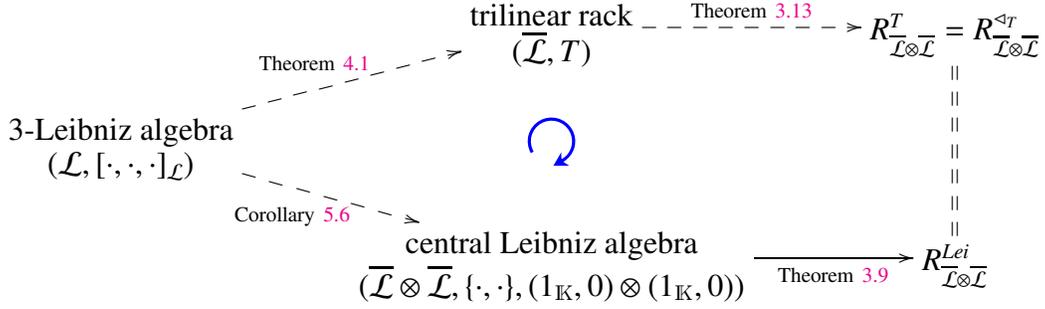
\vspace{-0.5cm}
\begin{center}
\begin{displaymath}
\xymatrix@C=7.5ex@R=2ex{
  &\txt{{\rm trilinear rack}~\\$(\overline{\huaL},T)$}
  \ar@[blue]@{-->}[r]^-{{\rm Theorem}~\ref{TSD-to-quan-rack}}
  &\txt{$R^{T}_{\overline{\huaL}\otimes\overline{\huaL}}=R^{\lhd_T}_{\overline{\huaL}\otimes\overline{\huaL}}$}
  \ar@[blue]@{==}[dd]\\
  \txt{{\rm $3$-Leibniz algebra}\\$(\huaL,[\cdot,\cdot,\cdot]_\huaL)$}
  \ar@[blue]@{-->}[dr]_-{{\rm Corollary}~\ref{3-lei-cen-ext-to-sol}\quad}
  \ar@[blue]@{-->}[ur]^-{{\rm Theorem}~\ref{3-Lei-to-TSD}\quad}
  &\rotatebox{165}{\color{blue}{\txt{\Huge $\circlearrowright$}}}&\\
  &\txt{{\rm central Leibniz algebra~}\\
  $(\overline{\huaL}\otimes\overline{\huaL},\{\cdot,\cdot\},(1_\mathds{K},0)\otimes(1_\mathds{K},0))$}
  \ar[r]_-{\rm Theorem~\ref{lei to solution}}
  &\txt{$R^{Lei}_{\overline{\huaL}\otimes\overline{\huaL}}$}
}
\end{displaymath}
\end{center}

There is another way to construct a solution of the Yang-Baxter equation from the given $3$-Leibniz algebra $(\huaL,[\cdot,\cdot,\cdot]_\huaL)$. We consider a special central extension $\K\oplus(\huaL\otimes\huaL)$ of the Leibniz algebra$(\huaL\otimes\huaL,\{\cdot,\cdot\})$, where $\{\cdot,\cdot\}$ is defined by \eqref{fundamental-element-formula}, then we obtain a solution of the Yang-Baxter equation by Theorem \ref{lei to solution}.

\begin{cor}\label{3-lei-to-lei-and-sol-1}
Let $(\huaL,[\cdot,\cdot,\cdot]_\huaL)$ be a $3$-Leibniz algebra over $\K$.
Then we obtain a solution of the Yang-Baxter equation $R^{Lei}_{\mathds{K}\oplus(\huaL\otimes\huaL)}$ on $\mathds{K}\oplus(\huaL\otimes\huaL)$ given by
\begin{eqnarray*}
&&R^{Lei}_{\mathds{K}\oplus(\huaL\otimes\huaL)}
\Big((a,x_1\otimes x_2)\otimes(b,y_1\otimes y_2)\Big)\\
&=&(b,y_1\otimes y_2)\otimes(a,x_1\otimes x_2)+(1_\mathds{K},0)\otimes(0,[x_1,y_1,y_2]_\huaL\otimes x_2+x_1\otimes[x_2,y_1,y_2]_\huaL).
\end{eqnarray*}
\end{cor}
\begin{proof}
By Proposition \ref{3-lei-to-lei}, $(\huaL\otimes\huaL,\{\cdot,\cdot\})$ is a Leibniz algebra. Then $(\K\oplus(\huaL\otimes\huaL),[\cdot,\cdot]_0,(1_\K,0))$ is a central extension of the Leibniz algebra $\huaL\otimes\huaL$, where $[\cdot,\cdot]_0$ is given by \eqref{cen-ext-bracket} with $\omega=0$, that is, for any $(a,x_1\otimes x_2)$, $(b,y_1\otimes y_2)\in\mathds{K}\oplus(\huaL\otimes\huaL)$,
\begin{eqnarray*}
&&[(a,x_1\otimes x_2),(b,y_1\otimes y_2)]_0=(0,\{x_1\otimes x_2,y_1\otimes y_2\})
=(0,[x_1,y_1,y_2]_\huaL\otimes x_2+x_1\otimes[x_2,y_1,y_2]_\huaL).
\end{eqnarray*}
Then by Theorem \ref{lei to solution}, there is a solution of the Yang-Baxter equation as follows:
\begin{eqnarray*}
\quad\quad&&R^{Lei}_{\mathds{K}\oplus(\huaL\otimes\huaL)}
\Big((a,x_1\otimes x_2)\otimes(b,y_1\otimes y_2)\Big)\\
&=&(b,y_1\otimes y_2)\otimes (a,x_1\otimes x_2)
+(1_\K,0)\otimes[(a,x_1\otimes(b,y_1\otimes y_2)),y]_0\\
&=&(b,y_1\otimes y_2)\otimes (a,x_1\otimes x_2)
+(1_\K,0)\otimes(0,[x_1,y_1,y_2]_\huaL\otimes x_2+x_1\otimes[x_2,y_1,y_2]_\huaL).\quad\quad\qedhere
\end{eqnarray*}
\end{proof}

Define an injective map
$\frks:\mathds{K}\oplus(\huaL\otimes\huaL)\rightarrow\overline{\huaL}\otimes\overline{\huaL}$ as follows:
\begin{eqnarray*}
\frks(a,x\otimes y)=a(1_\mathds{K},0)\otimes(1_\mathds{K},0)+(0,x)\otimes(0,y),\,\,\,\forall (a,x\otimes y)\in\mathds{K}\oplus(\huaL\otimes\huaL).
\end{eqnarray*}
For any $(a,x_1\otimes x_2)$, $(b,y_1\otimes y_2)\in\mathds{K}\oplus(\huaL\otimes\huaL)$, we have
\begin{eqnarray*}
&&\frks[(a,x_1\otimes x_2),(b,y_1\otimes y_2)]_0\\
&=&\frks(0,[x_1,y_1,y_2]_\huaL\otimes x_2+x_1\otimes[x_2,y_1,y_2]_\huaL)\\
&=&(0,[x_1,y_1,y_2]_\huaL)\otimes(0,x_2)+(0,x_1)\otimes(0,[x_2,y_1,y_2]_\huaL),\\
&&\{\frks(a,x_1\otimes x_2),\frks(b,y_1\otimes y_2)\}\\
&=&\{a(1_\mathds{K},0)\otimes(1_\mathds{K},0)+(0,x_1)\otimes(0,x_2),b(1_\mathds{K},0)\otimes(1_\mathds{K},0)+(0,y_1)\otimes(0,y_2)\}\\
&=&[(0,x_1),(0,y_1),(0,y_2)]_{\overline{\huaL}}\otimes(0,x_2)+(0,x_1)\otimes[(0,x_2),(0,y_1),(0,y_2)]_{\overline{\huaL}}\\
&=&(0,[x_1,y_1,y_2]_\huaL)\otimes(0,x_2)+(0,x_1)\otimes(0,[x_2,y_1,y_2]_\huaL),
\end{eqnarray*}
which implies that $\frks[(a,x_1\otimes x_2),(b,y_1\otimes y_2)]_0=\{\frks(a,x_1\otimes x_2),\frks(b,y_1\otimes y_2)\}$. Therefore, $(\mathds{K}\oplus(\huaL\otimes\huaL),[\cdot,\cdot]_0)$ is a subalgebra of $(\overline{\huaL}\otimes\overline{\huaL},\{\cdot,\cdot\})$.

\emptycomment{
\begin{lem}
Let $(\huaL,[\cdot,\cdot,\cdot]_\huaL)$ be a $3$-Leibniz algebra over $\mathds{K}$, then $(\mathds{K}\oplus(\huaL\otimes\huaL),[\cdot,\cdot]_0)$ is a subalgebra of $((\mathds{K}\oplus\huaL)\otimes(\mathds{K}\oplus\huaL),\{\cdot,\cdot\})$.
\end{lem}

\begin{proof}
Define an injective map
$\frks:\mathds{K}\oplus(\huaL\otimes\huaL)\rightarrow(\mathds{K}\oplus\huaL)\otimes(\mathds{K}\oplus\huaL)$ by
\begin{eqnarray*}
\frks(a,x\otimes y)=a(1_\mathds{K},0)\otimes(1_\mathds{K},0)+(0,x)\otimes(0,y).
\end{eqnarray*}
Then for any $(a,x_1\otimes x_2),(b,y_1\otimes y_2)\in\mathds{K}\oplus(\huaL\otimes\huaL)$, we have
\begin{eqnarray*}
\frks[(a,x_1\otimes x_2),(b,y_1\otimes y_2)]_0
&=&\frks(0,\{x_1\otimes x_2,y_1\otimes y_2\})\\
&=&\frks(0,[x_1,y_1,y_2]_\huaL\otimes x_2+x_1\otimes[x_2,y_1,y_2]_\huaL)\\
&=&(0,[x_1,y_1,y_2]_\huaL)\otimes(0,x_2)+(0,x_1)\otimes(0,[x_2,y_1,y_2]_\huaL),\\
\{\frks(a,x_1\otimes x_2),\frks(b,y_1\otimes y_2)\}
&=&\{a(1_\mathds{K},0)\otimes(1_\mathds{K},0)+(0,x_1)\otimes(0,x_2),
b(1_\mathds{K},0)\otimes(1_\mathds{K},0)+(0,y_1)\otimes(0,y_2)\}\\
&=&(0,[x_1,y_1,y_2]_\huaL)\otimes(0,x_2)+(0,x_1)\otimes(0,[x_2,y_1,y_2]_\huaL),
\end{eqnarray*}
which implies that $\frks[(a,x_1\otimes x_2),(b,y_1\otimes y_2)]_0=\{\frks(a,x_1\otimes x_2),\frks(b,y_1\otimes y_2)\}$. Thus, $(\mathds{K}\oplus(\huaL\otimes\huaL),[\cdot,\cdot]_0)$ is a subalgebra of $(\mathds{K}\oplus\huaL)\otimes(\mathds{K}\oplus\huaL)$.
\end{proof}
}

\begin{pro}
With the above notation, $\frks$ is a homomorphism from $R^{Lei}_{\mathds{K}\oplus(\huaL\otimes\huaL)}$ to $R^{Lei}_{\overline{\huaL}\otimes\overline{\huaL}}$, such that we have the commutative diagram \eqref{diagram:sec5}.
\end{pro}
\begin{proof}
By direct calculation, we have the following commutative diagram:
\vspace{-0.5cm}
\begin{center}
\begin{displaymath}
\xymatrix@C=8ex{
  \txt{$\Big(\mathds{K}\oplus(\huaL\otimes\huaL)\Big)\otimes\Big(\mathds{K}\oplus(\huaL\otimes\huaL)\Big)$}
  \ar@{^{(}-->}[d]_-{\frks\otimes\frks} \ar[r]^-{R^{Lei}_{\mathds{K}\oplus(\huaL\otimes\huaL)}}
  &  \txt{$\Big(\mathds{K}\oplus(\huaL\otimes\huaL)\Big)\otimes\Big(\mathds{K}\oplus(\huaL\otimes\huaL)\Big)$}
  \ar@{^{(}-->}[d]^-{\frks\otimes\frks} \\
   \txt{$\Big(\overline{\huaL}\otimes\overline{\huaL}\Big)\otimes\Big(\overline{\huaL}\otimes\overline{\huaL}\Big)$}
  \ar[r]_-{R^{Lei}_{\overline{\huaL}\otimes\overline{\huaL}}}
  & \txt{$\Big(\overline{\huaL}\otimes\overline{\huaL}\Big)\otimes\Big(\overline{\huaL}\otimes\overline{\huaL}\Big)$}
}
\end{displaymath}
\end{center}
which means that $\frks$ is a homomorphism from $R^{Lei}_{\mathds{K}\oplus(\huaL\otimes\huaL)}$ to $R^{Lei}_{\overline{\huaL}\otimes\overline{\huaL}}$.
\end{proof}
\emptycomment{
For central Leibniz algebras $(\mathds{K}\oplus(\huaL\otimes\huaL),[\cdot,\cdot]_0,(1_\mathds{K},0))$ and $((\mathds{K}\oplus\huaL)\otimes(\mathds{K}\oplus\huaL),\{\cdot,\cdot\},(1_\mathds{K},0)\otimes(1_\mathds{K},0))$, by Proposition \ref{cen-Lei-to-quan-rack}, we obtain respectively linear racks $(\mathds{K}\oplus(\huaL\otimes\huaL),\lhd_{\mathds{K}\oplus(\huaL\otimes\huaL)})$ and $((\mathds{K}\oplus\huaL)\otimes(\mathds{K}\oplus\huaL),\lhd_{(\mathds{K}\oplus\huaL)\otimes(\mathds{K}\oplus\huaL)})$ as follows:
\begin{eqnarray*}
&&(a,x_1\otimes x_2)\lhd_{\mathds{K}\oplus(\huaL\otimes\huaL)}(b,y_1\otimes y_2)\\
&=&(ab,bx_1\otimes x_2+[x_1,y_1,y_2]_\huaL\otimes x_2+x_1\otimes[x_2,y_1,y_2]_\huaL),\\
&&(a_1,x_1)\otimes(a_2,x_2)\lhd_{(\mathds{K}\oplus\huaL)\otimes(\mathds{K}\oplus\huaL)}(b_1,x_1)\otimes(b_2,x_2)\\
&=&b_1 b_2(a_1,x_1)\otimes(a_2,x_2)+(0,[x_1,y_1,y_2]_\huaL)\otimes(a_2,x_2)+(a_1,x_1)\otimes(0,[x_2,y_1,y_2]_\huaL).
\end{eqnarray*}}

\begin{ex}
Consider the $3$-Leibniz algebra $(\BO,[\cdot,\cdot,\cdot])$ showed in {\rm Example \ref{ex-octonion}}. By {\rm Corollary \ref{3-lei-cen-ext-to-sol}}, we obtain a solution of the Yang-Baxter equation on $(\R\oplus\BO)\otimes(\R\oplus\BO)$ as follows:
\begin{eqnarray*}
&&R^{Lei}_{(\R\oplus\BO)\oplus(\R\oplus\BO)}
\Big((a_1,x_1)\otimes(a_2,x_2)\otimes(b_1,y_1)\otimes(b_2,y_2)\Big)\\
&=&(b_1,y_1)\otimes(b_2,y_2)\otimes(a_1,x_1)\otimes(a_2,x_2)\\
&&+(1,0)\otimes(1,0)\otimes(0,y_2(y_1x_1)-y_1(y_2x_1)+(x_1y_1)y_2-(x_1y_2)y_1+(y_1x_1)y_2-y_1(x_1y_2))\otimes(a_2,x_2)\\
&&+(1,0)\otimes(1,0)\otimes(a_1,x_1)\otimes(0,y_2(y_1x_2)-y_1(y_2x_2)+(x_2y_1)y_2-(x_2y_2)y_1+(y_1x_2)y_2-y_1(x_2y_2)).
\end{eqnarray*}
By {\rm Corollary \ref{3-lei-to-lei-and-sol-1}}, we obtain a solution of the Yang-Baxter equation on $\R\oplus(\BO\otimes\BO)$ as follows:
\begin{eqnarray*}
&&R^{Lei}_{\R\oplus(\BO\otimes\BO)}
\Big((a,x_1\otimes x_2)\otimes(b,y_1\otimes y_2)\Big)\\
&=&(b,y_1\otimes y_2)\otimes(a,x_1\otimes x_2)\\
&&+(1,0)\otimes\Big(0,\big(y_2(y_1x_1)-y_1(y_2x_1)+(x_1y_1)y_2-(x_1y_2)y_1+(y_1x_1)y_2-y_1(x_1y_2)\big)\otimes x_2\Big)\\
&&+(1,0)\otimes\Big(x_1\otimes \big(y_2(y_1x_2)-y_1(y_2x_2)+(x_2y_1)y_2-(x_2y_2)y_1+(y_1x_2)y_2-y_1(x_2y_2)\big)\Big).
\end{eqnarray*}
\end{ex}

\begin{ex}
Let $\huaL$ be a $2$-dimensional vector space with a basis $\{e_1,e_2\}$.
Then $\huaL$ equipped with the following linear map $[\cdot,\cdot,\cdot]:\huaL\otimes\huaL\otimes\huaL\to\huaL$ is a $3$-Leibniz algebra:
$$[e_1,e_1,e_2]=e_2=-[e_1,e_2,e_1].$$
Denote by $\frke_1=e_1\otimes e_1$, $\frke_2=e_1\otimes e_2$, $\frke_3=e_2\otimes e_1$ and $\frke_4=e_2\otimes e_2$. Then we obtain a basis of $\Big(\K\oplus(\huaL\otimes\huaL)\Big)\otimes\Big(\K\oplus(\huaL\otimes\huaL)\Big)$ as follows:
$$\{(1_\K,0)\otimes(1_\K,0),(1_\K,0)\otimes(0,\frke_1),(1_\K,0)\otimes(0,\frke_2),
\cdots,(0,\frke_4)\otimes(0,\frke_3),(0,\frke_4)\otimes(0,\frke_4)\}.$$
By {\rm Corollary \ref{3-lei-to-lei-and-sol-1}}, we obtain a solution of the Yang-Baxter equation $R^{Lei}_{\mathds{K}\oplus(\huaL\otimes\huaL)}$ on $\mathds{K}\oplus(\huaL\otimes\huaL)$ with respect to the above basis as follows:
\begin{equation*}
R^{Lei}_{\mathds{K}\oplus(\huaL\otimes\huaL)}={\scriptsize
\left(
  \begin{array}{ccccccccccccccccccccccccc}
    1 & 0 & 0 & 0 & 0 & 0 & 0 & 0 & 0 & 0 & 0 & 0 & 0 & 0 & 0 & 0 & 0 & 0 & 0 & 0 & 0 & 0 & 0 & 0 & 0 \\
    0 & 0 & 0 & 0 & 0 & 1 & 0 & 0 & 0 & 0 & 0 & 0 & 0 & 0 & 0 & 0 & 0 & 0 & 0 & 0 & 0 & 0 & 0 & 0 & 0 \\
    0 & 0 & 0 & 0 & 0 & 0 & 0 & 1 & -1 & 0 & 1 & 0 & 0 & 0 & 0 & 0 & 0 & 0 & 0 & 0 & 0 & 0 & 0 & 0 & 0 \\
    0 & 0 & 0 & 0 & 0 & 0 & 0 & 1 & -1 & 0 & 0 & 0 & 0 & 0 & 0 & 1 & 0 & 0 & 0 & 0 & 0 & 0 & 0 & 0 & 0 \\
    0 & 0 & 0 & 0 & 0 & 0 & 0 & 0 & 0 & 0 & 0 & 0 & 1 & -1 & 0 & 0 & 0 & 1 & -1 & 0 & 1 & 0 & 0 & 0 & 0  \\
    0 & 1 & 0 & 0 & 0 & 0 & 0 & 0 & 0 & 0 & 0 & 0 & 0 & 0 & 0 & 0 & 0 & 0 & 0 & 0 & 0 & 0 & 0 & 0 & 0 \\
    0 & 0 & 0 & 0 & 0 & 0 & 1 & 0 & 0 & 0 & 0 & 0 & 0 & 0 & 0 & 0 & 0 & 0 & 0 & 0 & 0 & 0 & 0 & 0 & 0 \\
    0 & 0 & 0 & 0 & 0 & 0 & 0 & 0 & 0 & 0 & 0 & 1 & 0 & 0 & 0 & 0 & 0 & 0 & 0 & 0 & 0 & 0 & 0 & 0 & 0 \\
    0 & 0 & 0 & 0 & 0 & 0 & 0 & 0 & 0 & 0 & 0 & 0 & 0 & 0 & 0 & 0 & 1 & 0 & 0 & 0 & 0 & 0 & 0 & 0 & 0 \\
    0 & 0 & 0 & 0 & 0 & 0 & 0 & 0 & 0 & 0 & 0 & 0 & 0 & 0 & 0 & 0 & 0 & 0 & 0 & 0 & 0 & 1 & 0 & 0 & 0 \\
    0 & 0 & 1 & 0 & 0 & 0 & 0 & 0 & 0 & 0 & 0 & 0 & 0 & 0 & 0 & 0 & 0 & 0 & 0 & 0 & 0 & 0 & 0 & 0 & 0 \\
    0 & 0 & 0 & 0 & 0 & 0 & 0 & 1 & 0 & 0 & 0 & 0 & 0 & 0 & 0 & 0 & 0 & 0 & 0 & 0 & 0 & 0 & 0 & 0 & 0 \\
    0 & 0 & 0 & 0 & 0 & 0 & 0 & 0 & 0 & 0 & 0 & 0 & 1 & 0 & 0 & 0 & 0 & 0 & 0 & 0 & 0 & 0 & 0 & 0 & 0 \\
    0 & 0 & 0 & 0 & 0 & 0 & 0 & 0 & 0 & 0 & 0 & 0 & 0 & 0 & 0 & 0 & 0 & 1 & 0 & 0 & 0 & 0 & 0 & 0 & 0 \\
    0 & 0 & 0 & 0 & 0 & 0 & 0 & 0 & 0 & 0 & 0 & 0 & 0 & 0 & 0 & 0 & 0 & 0 & 0 & 0 & 0 & 0 & 1 & 0 & 0 \\
    0 & 0 & 0 & 1 & 0 & 0 & 0 & 0 & 0 & 0 & 0 & 0 & 0 & 0 & 0 & 0 & 0 & 0 & 0 & 0 & 0 & 0 & 0 & 0 & 0 \\
    0 & 0 & 0 & 0 & 0 & 0 & 0 & 0 & 1 & 0 & 0 & 0 & 0 & 0 & 0 & 0 & 0 & 0 & 0 & 0 & 0 & 0 & 0 & 0 & 0 \\
    0 & 0 & 0 & 0 & 0 & 0 & 0 & 0 & 0 & 0 & 0 & 0 & 0 & 0 & 0 & 0 & 0 & 0 & 0 & 0 & 0 & 0 & 0 & 0 & 0 \\
    0 & 0 & 0 & 0 & 0 & 0 & 0 & 0 & 0 & 0 & 0 & 0 & 0 & 1 & 0 & 0 & 0 & 0 & 1 & 0 & 0 & 0 & 0 & 0 & 0 \\
    0 & 0 & 0 & 0 & 0 & 0 & 0 & 0 & 0 & 0 & 0 & 0 & 0 & 0 & 0 & 0 & 0 & 0 & 0 & 0 & 0 & 0 & 0 & 1 & 0 \\
    0 & 0 & 0 & 0 & 1 & 0 & 0 & 0 & 0 & 0 & 0 & 0 & 0 & 0 & 0 & 0 & 0 & 0 & 0 & 0 & 0 & 0 & 0 & 0 & 0 \\
    0 & 0 & 0 & 0 & 0 & 0 & 0 & 0 & 0 & 1 & 0 & 0 & 0 & 0 & 0 & 0 & 0 & 0 & 0 & 0 & 0 & 0 & 0 & 0 & 0 \\
    0 & 0 & 0 & 0 & 0 & 0 & 0 & 0 & 0 & 0 & 0 & 0 & 0 & 0 & 1 & 0 & 0 & 0 & 0 & 0 & 0 & 0 & 0 & 0 & 0  \\
    0 & 0 & 0 & 0 & 0 & 0 & 0 & 0 & 0 & 0 & 0 & 0 & 0 & 0 & 0 & 0 & 0 & 0 & 0 & 1 & 0 & 0 & 0 & 0 & 0 \\
    0 & 0 & 0 & 0 & 0 & 0 & 0 & 0 & 0 & 0 & 0 & 0 & 0 & 0 & 0 & 0 & 0 & 0 & 0 & 0 & 0 & 0 & 0 & 0 & 1 \\
  \end{array}
\right)}.
\end{equation*}
\end{ex}

\noindent
{\bf Acknowledgements. } This research is supported by NSFC (12471060).


\begin{thebibliography}{a}

\bibitem{Abramov}
V. Abramov and E. Zappala, 3-Lie algebras, ternary Nambu-Lie algebras and the Yang-Baxter equation. \emph{J. Geom. Phys.} {\bf 183} (2023), 104687, 19 pp.

\bibitem{ABRW}
C. Alexandre, M. Bordemann, S. Rivi${\rm\grave{e}}$re and F. Wagemann,
Structure theory of rack-bialgebras. \emph{J. Gen. Lie Theory Appl.} {\bf 10} (2016),  1000244, 20 pp.

\bibitem{AG}
N. Andruskiewitsch and M. Gra\~na,
From racks to pointed Hopf algebras. \emph{Adv. Math.} {\bf 178} (2003), 177-243.

\bibitem{AOR}
S. Ayupov, B. Omirov and I. Rakhimov, Leibniz algebras: structure and classification. CRC Press, Boca Raton, FL (2020).

\bibitem{BGST}
C. Bai, L. Guo, Y. Sheng and R. Tang,
Post-groups, (Lie-)Butcher groups and the Yang-Baxter equation. \emph{Math. Ann.} {\bf 388} (2024), 3127-3167.

\bibitem{BZ}
R. Bai and J. Zhang,
The classification of nilpotent Leibniz 3-algebras.
\emph{Acta Math. Sci. Ser. B} {\bf 31} (2011), 1997-2006.

\bibitem{BG}
V. G. Bardakov and V. Gubarev,
Rota-Baxter groups, skew left braces, and the Yang-Baxter equation. \emph{J. Algebra} {\bf 596} (2022), 328-351.

\bibitem{Baxter}
R. J. Baxter, Partition function of the eight-vertex lattice model. \emph{Ann. Phys.} {\bf 70} (1972), 193-228.


\bibitem{Biyogmam}
G. R. Biyogmam, Lie $n$-racks. \emph{C. R. Math. Acad. Sci. Paris} {\bf 349} (2011), 957-960.



\bibitem{Carter}
J. Carter, A. Crans, M. Elhamdadi and M. Saito,
Cohomology of categorical self-distributivity.
\emph{J. Homotopy Relat. Struct.} {\bf 3} (2008), 13-63.

\bibitem{Casas}
J. M. Casas, J.-L. Loday and T. Pirashvili,
Leibniz $n$-algebras. \emph{Forum Math.} {\bf 14} (2002), 189-207.


\bibitem{Drinfeld}
V. G. Drinfel'd,
On some unsolved problems in quantum group theory. \emph{Lecture Notes in Math.}
{\bf 1510} (1992), 1-8.

\bibitem{EGM}
M. Elhamdadi, M. Green and A. Makhlouf, Ternary distributive structures and quandles.
\emph{Kyungpook Math. J.} {\bf 56} (2016), 1-27.

\bibitem{ESZ}
M. Elhamdadi, M. Saito and E. Zappala,
Higher arity self-distributive operations in Cascades and their cohomology.
\emph{J. Algebra Appl.} {\bf 20} (2021), 2150116, 33 pp.


\bibitem{ESS}
P. Etingof, T. Schedler and A. Soloviev,
Set-theoretical solutions to the quantum Yang-Baxter equation. \emph{Duke Math. J.} {\bf 100} (1999), 169-209.

\bibitem{FRT}
L. D. Faddeev, N. Yu. Reshetikhin and L. A. Takhtajan,
Quantization of Lie groups and Lie algebras. \emph{Academic Press, Inc., Boston, MA} {\bf 1} (1988), 129-139.

\bibitem{FR}
R. Fenn and C. Rourke,
Racks and links in codimension two. \emph{J. Knot Theory Ramifications} {\bf 1} (1992), 343-406.




\bibitem{GV2}
L. Guarnieri and L. Vendramin,
Skew braces and the Yang-Baxter equation. \emph{Math. Comp.} {\bf 86} (2017), 2519-2534.

\bibitem{GLS}
L. Guo, H. Lang and Y. Sheng,
Integration and geometrization of Rota-Baxter Lie algebras. \emph{Adv. Math.} {\bf 387} (2021), 107834, 34 pp.

\bibitem{Jones}
V. F. R. Jones,
Baxterization. \emph{Internat. J. Modern Phys. A} {\bf 6} (1991), 2035-2043.

\bibitem{Kinyon}
M. Kinyon, Leibniz algebras, Lie racks, and digroups. \emph{J. Lie Theory} {\bf 17} (2007), 99-114.

\bibitem{Krahmer}
U. Kr\"ahmer and F. Wagemann, A universal enveloping algebra for cocommutative rack bialgebras. \emph{Forum Math.} {\bf 31} (2019), 1305-1315.




\bibitem{Lebed3}
V. Lebed, Categorical aspects of virtuality and self-distributivity. \emph{J. Knot Theory Ramifications} {\bf 22} (2013), 1350045, 32 pp.

\bibitem{Lebed1}
V. Lebed, Homologies of algebraic structures via braidings and quantum shuffles. \emph{J. Algebra} {\bf 391} (2013), 152-192.


\bibitem{LSW}
J. Liu, Y. Sheng and C. Wang,
Omni $n$-Lie algebras and linearization of higher analogues of Courant algebroids. \emph{Int. J. Geom. Methods Mod. Phys.} {\bf 14} (2017), 1750113, 18 pp.

\bibitem{Loday1}
J.-L. Loday
Une version non commutative des alg${\rm \grave{e}}$bres de Lie: les alg${\rm \grave{e}}$bres de Leibniz. \emph{Enseign. Math.} {\bf 39} (1993), 269-293.

\bibitem{Loday}
J.-L. Loday and T. Pirashvili,
Universal enveloping algebras of Leibniz algebras and (co)homology.
\emph{Math. Ann.} {\bf 296} (1993), 139-158.

\bibitem{LYZ}
J. Lu, M. Yan and Y. Zhu,
On the set-theoretical Yang-Baxter equation. \emph{Duke Math. J.} {\bf 104} (2000), 1-18.

\bibitem{Nambu}
Y. Nambu, Generalized Hamiltonian dynamics. \emph{Phys. Rev. D} {\bf 7} (1973), 2405-2412.

\bibitem{Rump}
W. Rump, Braces, radical rings, and the quantum Yang-Baxter equation. \emph{J. Algebra} {\bf 307} (2007), 153-170.



\bibitem{Alan}
A. Weinstein, Omni-Lie algebras. Microlocal analysis of the Schr\"odinger equation and related topics (Japanese) (Kyoto, 1999), \emph{S${\bar{u}}$rikaisekikenky${\bar{u}}$sho K${\bar{u}}$ky${\bar{u}}$roku}, 2000, 1176: 1-102.



\bibitem{Yam}
M. Yamazaki, Octonions, $G_2$ and generalized Lie $3$-algebras.
\emph{Phys Lett B} {\bf 670} (2008), 215-219.

\bibitem{Yang}
C. Yang, Some exact results for the many-body problem in one dimension with repulsive delta-function interaction, \emph{Phys. Rev. Lett.} {\bf 19} (1967), 1312-1315.


\end{thebibliography}
 \end{document}